\documentclass{article}

\usepackage{tikz}
\usepackage{amsfonts}
\usepackage{url,hyperref}
\usepackage{amsmath,amsthm,amsfonts,amssymb,mathrsfs}
\usepackage{commath}
\usepackage{mathtools}
\usepackage{thmtools}
\usepackage{thm-restate}
\usepackage{fullpage,complexity}
\usepackage{todonotes}
\usepackage{comment}
\usepackage{complexity}
\usepackage{MnSymbol}

\usetikzlibrary{graphs}
\usetikzlibrary{graphs.standard}

\declaretheorem[name=Theorem]{theorem}

\declaretheorem[name=Lemma, sibling=theorem]{lemma}
\declaretheorem[name=Proposition, sibling=theorem]{proposition}
\declaretheorem[name=Definition, sibling=theorem]{definition}
\declaretheorem[name=Corollary, sibling=theorem]{corollary}

\declaretheorem[name=Observation, sibling=theorem]{observation}
\declaretheorem[name=Problem]{problem}

\declaretheorem[name=Claim, sibling=theorem]{claim}

\title{Parameterized Shortest Path Reconfiguration}

\author{
Nicolas Bousquet\footnote{Univ. Lyon, LIRIS, CNRS, Université Claude Bernard Lyon 1, Villeurbanne, France. E-mail: \href{mailto:nicolas.bousquet@cnrs.fr}{nicolas.bousquet@cnrs.fr}.}
\and Kshitij Gajjar\footnote{International Institute of Information Technology Hyderabad (IIIT-H), Hyderabad, India. E-mail: \href{mailto:kshitij@iiit.ac.in}{kshitij@iiit.ac.in}.}  
\and Abhiruk Lahiri\footnote{Department of Computer Science, Heinrich Heine University, D\"{u}sseldorf, Germany. E-mail: \href{mailto:abhiruk@hhu.de}{abhiruk@hhu.de}.} 
\and Amer~E.~Mouawad\footnote{Department of Computer Science, American University of Beirut, Beirut, Lebanon. E-mail: \href{mailto:aa368@aub.edu.lb}{aa368@aub.edu.lb}.}
}

\date{}

\begin{document}
\maketitle

\begin{abstract}
An $st$-shortest path, or $st$-path for short, in a graph $G$ is a shortest (induced) path from $s$ to $t$ in~$G$. Two $st$-paths are said to be adjacent if they differ on exactly one vertex. A reconfiguration sequence between two $st$-paths $P$ and $Q$ is a sequence of adjacent $st$-paths starting from $P$ and ending at $Q$. Deciding whether there exists a reconfiguration sequence between two given $st$-paths is known to be \textsf{PSPACE}-complete, even on restricted classes of graphs such as graphs of bounded bandwidth (hence pathwidth). On the positive side, and rather surprisingly, the problem is polynomial-time solvable on planar graphs.
In this paper, we study the parameterized complexity of the \textsc{Shortest Path Reconfiguration} (SPR) problem. We show that SPR is \textsf{W[1]}-hard parameterized by $k + \ell$, even when restricted to graphs of bounded (constant) degeneracy; here $k$ denotes the number of edges on an $st$-path, and $\ell$ denotes the length of a reconfiguration sequence from $P$ to $Q$. We complement our hardness result by establishing the fixed-parameter tractability of SPR parameterized by $\ell$ and restricted to nowhere-dense classes of graphs. Additionally, we establish fixed-parameter tractability of SPR when parameterized by the treedepth, by the cluster-deletion number, or by the modular-width of the input graph.
\end{abstract}

\section{Introduction}
Many algorithmic questions can be posed as follows:
given the description of a system state and the description of a state we would ``prefer'' the system to be in, is it possible to transform the system
from its current state into a more desired one without ``breaking'' the system in the process? And if yes, how many steps are needed? Such problems naturally arise in the fields of mathematical puzzles, operational research, computational geometry~\cite{DBLP:journals/comgeo/LubiwP15}, bioinformatics,  and quantum computing~\cite{DBLP:conf/icalp/GharibianS15}. These questions received a substantial amount of attention under the so-called \emph{combinatorial reconfiguration framework} in the last decade.
We refer the reader to the surveys by van den Heuvel~\cite{DBLP:books/cu/p/Heuvel13}, Nishimura~\cite{DBLP:journals/algorithms/Nishimura18}, and Bousquet \textit{et al.}~\cite{BousquetMNS22+} for more background on combinatorial reconfiguration.

\subparagraph{Shortest path reconfiguration.}
In this work, we focus on the reconfiguration of $st$-shortest paths (or $st$-paths for short) in undirected, unweighted, simple graphs. It is well-known that one can easily find an $st$-path in a graph in polynomial time. In order to define the reconfiguration variant of the problem, we first require a notion of adjacency between $st$-paths.

As is common in the combinatorial reconfiguration framework, we focus on two models; the token-jumping model (TJ) and the token-sliding model (TS). We say that two $st$-paths are \emph{TJ-adjacent} if they differ on exactly one vertex, i.e., all the vertices are the same except at a unique position $p$. We say that two $st$-paths $P$ and $Q$ are \emph{TS-adjacent} if they are TJ-adjacent and the $p$\textsuperscript{th} vertex of $P$ and the $p$\textsuperscript{th} vertex of $Q$ are adjacent. A \emph{reconfiguration sequence from $P$ to $Q$} (if it exists) is a sequence of adjacent shortest paths starting at $P$ and ending at $Q$. In the \textsc{Shortest Path Reconfiguration} (SPR) problem, we are given a graph $G$, two vertices $s$ and $t$, two $st$-paths $P$ and $Q$ of length $k$ each, and the goal is to decide whether a reconfiguration sequence from $P$ to $Q$ exists. In the \textsc{Shortest Shortest Path Reconfiguration} (SSPR) problem, we are additionally given an integer $\ell$ which is an upper bound on the length of the desired reconfiguration sequence. 
Reconfiguration of shortest paths has many applications, e.g., in network design and operational research (we refer the interested reader to~\cite{GajjarJ0L22} for a detailed discussion around these applications). 

Many reconfiguration problems, SPR and SSPR included, naturally lie in the class \textsf{PSPACE}. Since there are no simple polynomial-time checkable certificates (as reconfiguration sequences are possibly of exponential length), they are generally not in \textsf{NP}. A decade ago, Bonsma~\cite{Bonsma13} proved that SPR is \textsf{PSPACE}-complete. In fact, the problem remains \textsf{PSPACE}-complete even when restricted to bipartite graphs~\cite{Bonsma13}, line graphs~\cite{GajjarJ0L22}, and graphs of bounded bandwidth/pathwidth/treewidth~\cite{Wrochna18}.
Several groups studied the complexity of the problem in other restricted graph classes such as grid graphs~\cite{AsplundEHHNW18}, claw-free graphs, chordal graphs~\cite{Bonsma13}, and circle graphs~\cite{GajjarJ0L22}. The most notable result has been obtained by Bonsma who showed that \textsc{Shortest Path Reconfiguration} can be decided in polynomial time for planar graphs~\cite{Bonsma17}. This result is rather surprising in the reconfiguration setting since most reconfiguration problems are known to be \textsf{PSPACE}-complete on planar graphs, see e.g.~\cite{DBLP:journals/tcs/ItoDHPSUU11, DBLP:journals/dam/ItoKD12, BousquetHKMS23}.

\subparagraph{Our results.}
Our focus is on the parameterized complexity of shortest path reconfiguration problems; which, to the best of our knowledge, has not been studied so far.  Other reconfiguration problems have been widely studied from a parameterized perspective in the last decade, see,  e.g.,~\cite{BousquetMNS22+} for a survey. 
A problem is \emph{fixed-parameter tractable}, \textsf{FPT} for short, on a class $\mathcal{C}$ of graphs with respect to a parameter $\kappa$, if there is an algorithm deciding whether a given input instance with graph $G \in \mathcal{C}$ admits a solution in time $f(\kappa) \cdot |V(G)|^c$, for a computable function~$f$ and constant $c$.  A \emph{kernelization algorithm} is a polynomial-time algorithm that reduces an input instance to an equivalent instance of size bounded in the parameter only (independent of the input size), known as a \emph{kernel}; we will say that two instances are \emph{equivalent} if they are both yes-instances or both no-instances.  Every fixed-parameter tractable problem admits a kernel, however, possibly of exponential or worse size. For efficient algorithms, it is therefore most desirable to obtain polynomial, or even linear, kernels. The \textsf{W}-hierarchy is a collection of parameterized complexity classes  $\textsf{FPT} \subseteq \textsf{W[1]} \subseteq \textsf{W[2]} \subseteq \cdots \subseteq \textsf{W[t]}$, for $t \in \mathbb{N}$. The conjecture $\textsf{FPT} \subsetneq \textsf{W[1]}$ can be seen as the analogue of the conjecture that $\textsf{P} \subsetneq \textsf{NP}$. Before stating our results precisely, let us formally define the problems we are interested in:

\medskip
\noindent
\textsc{Shortest Path Reconfiguration (SPR)} \\
\textbf{Input:} A graph $G$, two vertices $s,t$, 
 and two $st$-shortest paths $P,Q$ (each of length $k$). \\
\textbf{Question:} Is there a reconfiguration sequence from $P$ to $Q$?
\medskip

\medskip
\noindent
\textsc{Shortest Shortest Path Reconfiguration (SSPR)} \\
\textbf{Input:} A graph $G$, two vertices $s,t$, two $st$-shortest paths $P,Q$ (each of length $k$), and an integer $\ell$. \\
\textbf{Question:} Is there a reconfiguration sequence from $P$ to $Q$ of length at most $\ell$?
\medskip

In parameterized complexity, one is usually interested in two types of parameters; parameters related to the size of the solution or parameters related to the structure of the input graph.
For shortest path reconfiguration, there are two parameters related to the size of the solution which are the length $\ell$  of a reconfiguration sequence, and the length $k$ of the shortest $st$-paths (number of edges on the shortest $st$-paths) in $G$. Our first results will focus on these parameters. We will then discuss some parameters related to the graph structure such as treedepth and modular width. 
Our first result is a hardness result. We prove that the following holds (in both the token jumping and the token sliding model):

\begin{theorem}\label{thm:hardness1}
SPR is $\W[1]$-hard parameterized by $k$, and SSPR is $\W[1]$-hard parameterized by $k+\ell$.
\end{theorem}

We will prove Theorem~\ref{thm:hardness1} in Section~\ref{sec:hardness}.
The idea of the proof is a reduction from the \textsc{Multicolored Clique} problem. Let $(V_i)_{i \le k}$ be the vertices of an instance of the \textsc{Multicolored Clique} problem. Intuitively (the real proof being more technical), we will construct a graph where the length of the $st$-paths will be in $\mathcal{O}(k^2)$, each integer representing a vertex of the set $V_i$. 
The goal would be to transform a path $P$ into a path $Q$, forcing us to select a vertex in each set.
For every pair $i,j$, there exists an integer $r$ such that the $r$\textsuperscript{th} vertex corresponds to a vertex in $V_i$ and the $(r+1)$\textsuperscript{th} vertex corresponds to a vertex in $V_j$. The key argument of the proof consists in finding a mechanism to ensure that the vertex selected in each copy of $V_i$ is the same, which permits us to conclude that the subset of selected vertices is a multicolored clique of the desired size.

One can then naturally wonder if this hardness result can be pushed further. The answer is yes, and in fact, we prove that the problems are hard (in both the token jumping and the token sliding model) even restricted to a very simple class of graphs:

\begin{theorem}\label{thm:hardness2}
SPR is $\W[1]$-hard parameterized by $k$, and SSPR is $\W[1]$-hard parameterized by $k+\ell$, even when the inputs are restricted to graphs of constant degeneracy.
\end{theorem}

In order to prove Theorem~\ref{thm:hardness2}, we adapt the proof of Theorem~\ref{thm:hardness1} to appropriately reduce the degeneracy of the graph. We then complement these negative results with the following positive ones.

\begin{theorem}\label{thm:easy1}
SSPR is $\FPT$ parameterized by $\ell$ on nowhere-dense classes of graphs (in both the token jumping and the token sliding model).
\end{theorem}

The idea of the proof of Theorem~\ref{thm:easy1} consists in proving that if $k$ is too large compared to $\ell$ then there are many positions along the shortest paths that are already occupied by tokens that never have to move. Using this fact, we then contract parts of the paths in order to get $st$-paths of length $\mathcal{O}(f(\ell))$, for some computable function $f$. Now, since $k$ is bounded by some function of $\ell$, one can prove that the existence of a reconfiguration sequence of length $\ell$ can be verified via model checking a first-order formula $\phi$ whose size depends only on $\ell$. Combining this observation with the black-box result of~\cite{DBLP:conf/stoc/GroheKS14} that ensures that the model checking problem can be decided in time $\mathcal{O}(f(|\phi|) \cdot |V(G)|)$ on nowhere-dense graphs, we get the desired~result.

We proceed by considering some of the most commonly studied structural graph parameters. In particular, we prove the following:

\begin{theorem}\label{thm:easy2}
SPR and SSPR (in both the token jumping and the token sliding model) are $\FPT$ when parameterized by either the treedepth, the cluster deletion number, or the modular width of the input graph. 
\end{theorem}

To motivate the study of these parameters, we refer the reader to Figure~\ref{fig:classes} (formal definitions provided later). Recall that SPR is \textsf{PSPACE}-complete even when restricted to graphs of bounded bandwidth, pathwidth, treewidth, and cliquewidth~\cite{Wrochna18}. This implies para-\textsf{PSPACE}-hardness on the aforementioned classes. Hence, our Theorem~\ref{thm:easy2} almost completes the picture for structural parameterizations of the problems, leaving open the case of feedback vertex set number. 

\begin{figure}[ht]
    \centering
    \includegraphics[scale=0.7]{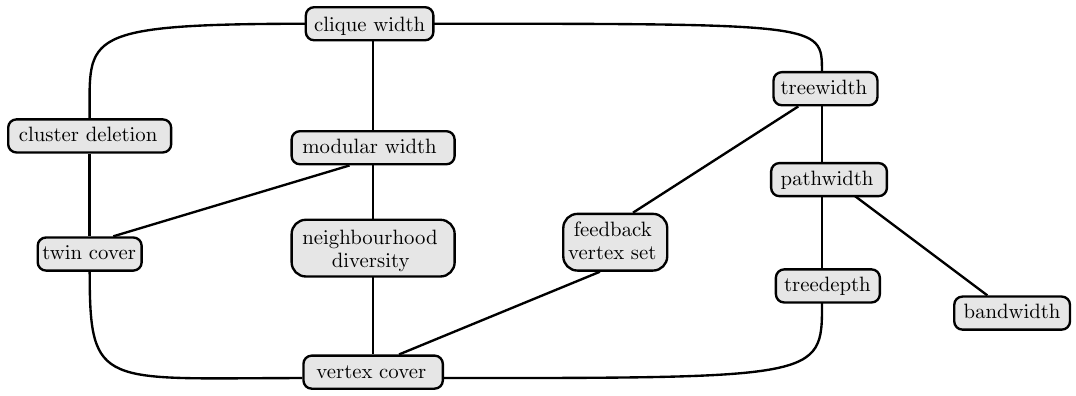}
    \caption{The graph parameters studied in this paper. A connection between two parameters
indicates the existence of a function in the one above that lower-bounds the one below.}
    \label{fig:classes}
\end{figure}

\subparagraph{Further discussions and open problems.} As we later show, it turns out that when solving the SPR problem parameterized by the feedback vertex set number of the graph, one can assume that $k$, the length of $st$-paths, is bounded linearly in the parameter. Hence, the following  remains an interesting open question:

\begin{problem}~\label{q:fvs}
Is SPR fixed-parameter tractable when parameterized by feedback vertex set number?
\end{problem}

When the feedback vertex set number is bounded, the graph can be seen as a disjoint union of trees plus a bounded number of additional vertices. One can easily remark that if vertices of the feedback vertex set are far apart in the $st$-paths then the structure is very rigid and very few tokens can move in the graph. However, when vertices of the feedback vertex set are close to one another (along the $st$-paths), there might exist some arbitrarily long paths between two layers in the layered partition of the graph. Here, the layered partition refers to the partitioning of the vertex set based on distance either from $s$ or from $t$. Tokens along these (layer) paths that do not belong to the feedback vertex set are not restricted and can traverse their corresponding layer path in both directions an unbounded number of times. In particular, it implies that, if there exists a reconfiguration sequence, that sequence might be arbitrarily long. So in order to design a reconfiguration sequence (from a kernelization perspective at least, which is known to be equivalent to fixed-parameter tractability), we have to find a way to reduce these long structures into structures of bounded length (see Figure~\ref{fig:hardinstances} for an example of such an instance). We were not able to solve this very special case of the problem.

\begin{figure}[ht]
    \centering
    \includegraphics[scale=0.4]{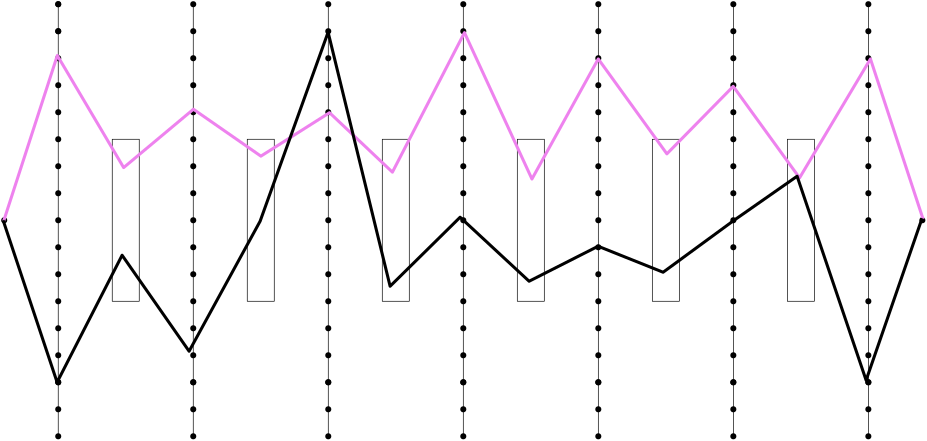}
    \caption{Example of a hard instance. Only edges of starting and ending paths are shown. Rectangles represent vertices of the feedback vertex set who neighborhood into the adjacent paths can be arbitrary.}
    \label{fig:hardinstances}
\end{figure}

As far as we know, it also remains an open question whether SPR is in \textsf{P} or is \textsf{NP}-complete on graphs of constant feedback vertex set number. Note that an \textsf{XP} algorithm follows immediately from the fact that (after appropriately discarding parts of the input) the number of $st$-paths is roughly $|V(G)|^f$, where $f$ denotes the feedback vertex set number.
Regardless, in case of a positive answer to Problem~\ref{q:fvs}, the next natural question is the following:

\begin{problem}
Is SPR fixed-parameter tractable when parameterized by $k$ on graphs of bounded pathwidth? What about bounded treewidth? How about parameterization by $k$ plus the treewidth? 
\end{problem}

It is an easy exercise to remark that SPR is \textsf{PSPACE}-complete on graphs of bounded bandwidth, pathwidth, and treewidth using a  reduction from \textsc{$H$-Word Reconfiguration}~\cite{Wrochna18}. 
When the treewidth is $1$, there exists a unique minimum $st$-path and the problem is simple. Trees and forests are graphs which are $1$-degenerate and every $1$-degenerate graph is a forest, however, the complexity of both SPR and SSPR remains open for $2$-degenerate graphs.


\begin{problem}
What is the complexity of SPR and SSPR on $2$-degenerate graphs? How about $3$-degenerate and $4$-degenerate graphs?
\end{problem}

\subparagraph{Related work.}
Reconfiguration of paths and other subgraphs has been considered before. In the example of paths, Demaine~\textit{et al.} proved in~\cite{DemaineEHJLUU19} that the problem of reconfiguring (arbitrary) paths is \textsf{PSPACE}-complete in general, and polynomial-time solvable for some restricted graph classes. When not restricted to shortest paths, the problem is quite different, since the extremities of the paths are not fixed and the goal is not necessarily to reconfigure shortest paths.

Reconfiguration problems on graphs of bounded feedback vertex set number and on graphs of bounded treewidth have already received a considerable amount of attention, and they are usually not easy to place in \textsf{FPT} (unlike their optimization counterparts, where a simple branching strategy or dynamic programming algorithm is usually enough to get an \textsf{FPT} algorithm). For instance, \textsc{Independent Set Reconfiguration} (in the token sliding model) on graphs of bounded feedback vertex set number is \textsf{FPT}; this fact follows easily from the multi-component reduction in~\cite{DBLP:journals/corr/abs-2204-05549}. However, the question is still open for the reconfiguration of dominating sets, for instance. The case of bounded treewidth graphs is open for both \textsc{Independent Set Reconfiguration} and \textsc{Dominating Set Reconfiguration} (in the sliding model)~\cite{BousquetMNS22+}.
\section{Hardness results}\label{sec:hardness}
We start with the case of SPR parameterized by $k$ on general graphs. The same reduction will imply the hardness of SSPR parameterized by $k + \ell$. We then describe how to modify the construction to obtain a graph of constant degeneracy. 
\subsection{General graphs}
Our reduction is from the \textsc{Regular Multicolored Clique} (RMC) problem, which is known to be \textsf{NP}-complete and \textsf{W[1]}-hard when parameterized by solution size $\kappa$~\cite{pietrzak2003parameterized}. The problem is defined as follows. We are given a $\kappa$-partite graph $G = (V,E)$ such that $V$ is partitioned into $\kappa$ independent sets $V = V_1\cupdot V_2\cupdot \cdots \cupdot V_\kappa$ and each partition has size exactly $n$, i.e., $|V| = \kappa n$. 
We denote the vertices of $V_i$ by $v_1^i,v_2^i,\ldots,v_{n}^i$. Moreover, every vertex $v_j^i \in V_i$ has exactly $r$ neighbors in every set $V_{i'}$, $i \neq i'$. In other words, every vertex in $G$ has degree exactly $r(\kappa - 1)$. Given an instance $(G,\kappa)$ of RMC, the goal is to decide if $G$ contains a clique of size $\kappa$, which we call a multicolored clique since it must contain exactly one vertex from each $V_i$, $i \in [\kappa]$. 
We reduce $(G,\kappa)$ to an instance $(G',s,t,P,Q)$ of SPR, where $P$ and $Q$ are $st$-paths in $G'$ of length $k = \mathcal{O}(\kappa^2)$.

\subparagraph{Properly colored $st$-paths.}
Before discussing $G'$, we start by describing a key gadget of our construction which is a graph called $H$. The graph $H$ consists of $\alpha = 6\kappa^2$ sets of vertices 
$H_1, H_2, \ldots, H_\alpha$ such that $|H_i| = n$ for each $i \in [\alpha]$. We group every three consecutive sets into $\beta = 2\kappa^2$ groups $R_1 = \{H_1, H_2, H_3\}$, $R_2 = \{H_4, H_5, H_6\}$, $R_3 = \{H_7, H_8, H_9\}$, $\ldots$, and $R_\beta = \{H_{\alpha - 2}, H_{\alpha - 1}, H_\alpha\}$. We call $H_i$ the $i$th layer of $H$ and $R_i$ the $i$th group of $H$; it will become clear later that a shortest path will select a vertex from each $H_i$. We also define a mapping $\mu: [\beta] \rightarrow [\kappa]$ such that each $R_i$ is mapped to some $V_j$, for $i \in [\beta]$ and $j \in [\kappa]$. In other words, each $R_i = \{H_a, H_b, H_c\}$ will  correspond to taking three copies of some $V_j$. We sometimes abuse notation and write $\mu(R_i) = V_j$ to denote the image of a set. We also overload notation and write $\mu(H_p) = V_j$ whenever $H_p \in R_i$ and $\mu(R_i) = V_j$. In other words, we also define a mapping $\mu: [\alpha] \rightarrow [\beta] \rightarrow [\kappa]$. 

Furthermore, we construct $\mu$ in such a way that, for every pair $(j,j')$, $j \neq j'$ and $j,j' \in [\kappa]$, there exists at least one integer $i < \beta$ such that $\mu(i) = j$, $\mu(i+1) = j'$. In other words, for every two sets $V_j$ and $V_{j'}$, there must exist two consecutive groups $R_i$ and $R_{i+1}$ such that $R_i$ is mapped to $V_j$ and $R_{i+1}$ is mapped to $V_{j'}$. One can easily check that it is indeed possible to construct such a function $\mu$ when $\beta = 2\kappa^2$. 
We define $\mu$ as follows:
\begin{align*}
    \text{For each $i \in [\beta]$, $R_i$ is mapped to $V_{\mu(i)}$, where }
    \displaystyle{
        \mu(i) = \begin{cases}
        1 + \lfloor (i-1)/2\kappa \rfloor &\text{if $i$ is odd};
        \\
        1 + ((i-2) \bmod 2\kappa)/2 &\text{if $i$ is even}.
    \end{cases}}
\end{align*}

\begin{observation}\label{obs:allpairs}
For each $(j,j')\in [\kappa]\times[\kappa]$ such that $j \neq j'$, there exists an $i\in[\beta - 1]$ such that $\mu(i) = j$ and $\mu(i+1) = j'$.
\end{observation}

We also define a mapping $\pi: R_i \rightarrow V_{\mu(i)}$ (and $\pi: H_i \rightarrow V_{\mu(i)}$) that maps every vertex of $R_i$ ($H_i$) to its corresponding vertex in $V_{\mu(i)}$. We note that each vertex of $V_{\mu(i)}$ appears three times in $R_i$ (once in each layer) and all three vertices map to the same vertex of $V_{\mu(i)}$. 
Let us now describe the edge set of $H$. For every $i \in [\beta]$, we add a matching between vertices of $H_j$ and $H_{j+1}$ and a matching between vertices of $H_{j+1}$ and $H_{j+2}$ whenever there exists a group $R_i$ such that $R_i = \{H_{j}, H_{j +1}, H_{j+2}\}$. For every two consecutive groups $R_i = \{H_{j}, H_{j+1}, H_{j+2}\}$ and $R_{i+1} = \{H_{j+3}, H_{j+4}, H_{j+5}\}$, we add in $H$ the edges of $G$ between $H_{j+2}$ and $H_{j+3}$. That is, we add between consecutive sets corresponding to different sets of $G$ the edges corresponding to the edges between those two sets in $G$. 
More formally, let $a \in H_{j+2}$, $b \in H_{j+3}$, $\pi(a) \in V_{\mu(i)}$, and $\pi(b) \in V_{\mu(i + i)}$. Then, there is an edge between vertices $a$ and $b$ in $H$ if and only if there is an edge between vertices $\pi(a)$ and $\pi(b)$ in $G$. 

Assume that we create a new graph $H'$ consisting of $H$ plus two additional vertices $s$ and $t$,  where $s$ is connected to all the vertices of $H_1$ and $t$ is connected to all the vertices of $H_{\alpha}$.
Note that any $st$-path in $H'$ must contain exactly one vertex from every layer. We say that an $st$-path $P$ is \emph{properly colored} whenever for any $a \in H_i$ and $b \in H_j$ (on the path) such that $\mu(i)=\mu(j)$, we have $\pi(a) = \pi(b)$. In other words, whenever two layers of $H$ (containing vertices of $P$) map to the same set of $V$ we must select the same vertices in both. We note that any $st$-path $P$ in $H'$ can intersect with a group $R_i$ in one of $n$ ways, i.e., the vertices of $P$ in $R_i$ all map to the same vertex of $V_{\mu(i)}$. 

\begin{observation}\label{obs:goodpath}
$H'$ contains a properly colored $st$-path $P$ (consisting of $6\kappa^2 + 2$ vertices) if and only if $G$ contains a multicolored clique of size $\kappa$. 
\end{observation}

\begin{proof}
The proof follows from the definition of a properly colored $st$-path in $H'$. Let $P$ be such a path. For path vertices $a \in H_i$ and $b \in H_j$ such that $\mu(i) = \mu(j)$ we have $\pi(a) = \pi(b)$. Moreover, for path vertices $a \in H_j$, $b \in H_{j + 1}$, and $c \in H_{j+2}$, where there exists $i \in [\beta]$ such that $R_i = \{H_j, H_{j + 1}, H_{j+2}\}$,  we have $\pi(a) = \pi(b) = \pi(c)$; recall that $\mu(R_i) = \mu(H_j) = \mu(H_{j + 1}) = \mu(H_{j+2})$. Hence, the internal vertices of $P$ correspond to exactly $\kappa$ distinct vertices of $G$, one vertex in each part of $V$. The fact that those vertices must form a multicolored clique in $G$ follows from Observation~\ref{obs:allpairs} and the fact that, within each group, we connect all three copies of a vertex by a path.

The converse follows by similar arguments. That is, suppose that $G$ contains a multicolored clique $C = \{v_1,\ldots,v_{\kappa}\}$ of size $\kappa$. Then, for every $i \le \beta$ such that $\mu(i)=j$, the vertices of $P$ in $R_i$ correspond to $\pi^{-1}(v_j)$. And, once these vertices are chosen, it is clear that they form a path $P$ in $H'$, since each pair of vertices in the multicolored clique is adjacent in the graph $G$. More formally, since we chose the same vertex in each layer of each group, we know that the chosen vertices within a group form a path. The fact that the paths within the groups connect to form an $st$-path $P$ is immediate from our construction and $C$ being a multicolored clique. 
\end{proof}

\subparagraph{Outline of the reduction.}
Assume that we add to the graph $H'$  two new (internally) vertex-disjoint $st$-paths $P$ and $Q$ each containing exactly $\alpha + 2$ vertices ($s$ and $t$ and one vertex per layer of $H$). We add all the edges between the $i$-th vertex of $P$ and the vertices in layers $i$, $i-1$, and $i+1$ of $H$ (with the assumption that $H_0 = \{s\}$ and $H_{\alpha + 1} = \{t\}$). Similarly, we add all the edges between the $i$-th vertex of $Q$ and the vertices in layers $i$, $i-1$, and $i+1$ of $H$. We denote the resulting graph by $H' + P + Q$.

Consider the instance $(H' + P + Q, s, t, P, Q)$ of SPR. If there exists a multicolored clique in $G$ then there exists a properly colored $st$-path in $H'$ by Observation~\ref{obs:goodpath}. By the definition of the edge set, on can easily see that we can transform $P$ into $Q$ by first moving the vertices of $P$ onto a properly colored $st$-path in $H'$ and then moving all the vertices to $Q$ one by one. Unfortunately, the converse is not necessarily true since we might not be consistent in the selection of vertices in $H'$. In other words, i.e. we have no reason to select a properly colored path (said differently, we might select vertices $a \in H_i$ and $b \in H_j$ in the path such that $\mu(i) = \mu(j)$, $\pi(a) \ne \pi(b)$, and $H_i$ and $H_j$ belong to different groups).

By considerably complicating the gadgetry, we will prove that we can handle this issue. To do so, we create a new gadget that will force us to select the same vertex for a fixed value of the image of $\mu$. We replicate our gadget to enforce the consistency of all the images of $\mu$. In addition to enforcing consistent selection of vertices, our construction further guarantees that choices cannot be undone.

Another issue in the simplistic construction of $H'$ described above is that we implicitly assume that we move from $P$ to a path fully contained in $H$ before going to $Q$. But nothing prevents an $st$-path to contain some vertices of $P$, then some vertices from $H$. then some vertices from $Q$, then more vertices from $H$, and so on. To avoid this phenomena, we shall add what we call  buffer space. We formalize all these ideas next.

\subparagraph{Buffers and collapses.} 
Most of the time, we will consider matchings and edges between sets of size $n$. Given two sets of size $n$ (with an implicit ordering), we define the natural matching as the matching that matches the vertices in increasing index order (in the natural way). We will sometimes consider edges between a set $A$ of size $n$ and a set $B$ of size larger than $n$ with a canonical mapping function to $\{1,\ldots,n \}$. By abuse of notation, we still denote by the natural matching the set of edges (that is not a matching anymore) that links the $i$-th vertex of $A$ and all the vertices that map to $i$~in~$B$.

We denote by $I_n$ and $J_n$ independent sets on $n$ vertices\footnote{These two notations for the same grapph will permit to simply description of the constructions in the rest of the proof.}. We let $\mathcal{I}^q$ and $\mathcal{J}^q$ denote the graphs obtained by taking $q$ copies of $I_n$ (resp. $J_n$) where consecutive copies of $I_n$ ($J_n$) are linked with the natural matching. Note that $\mathcal{I}^q$ (resp. $\mathcal{J}^q$) consists of  exactly $n$ paths on $q$ vertices. 

Let $R=R_1,R_2,\ldots,R_\gamma$ be a graph where edges are between consecutive sets and there is a canonical mapping from $R_1$ and $R_\gamma$ to $\{1,\ldots,n \}$ (in our proof, $R$ will be $H$ or a graph close to $H$).
We write $\Gamma(p,H,q) = \mathcal{I}^p \oplus H \oplus \mathcal{J}^q$ (or $\Gamma$ when $p,q,H$ are clear from context) to denote the graph obtained by taking a copy of $\mathcal{I}^p$, a copy of $\mathcal{J}^q$, a copy of $H$, and then adding the natural matching between the vertices of $I_q$ and $H_1$ as well as a matching between the vertices of $H_\alpha$ and $J_1$ (see Figure~\ref{fig:buffergadget}).  

\begin{figure}[ht]
    \centering
    \includegraphics[scale=0.3]{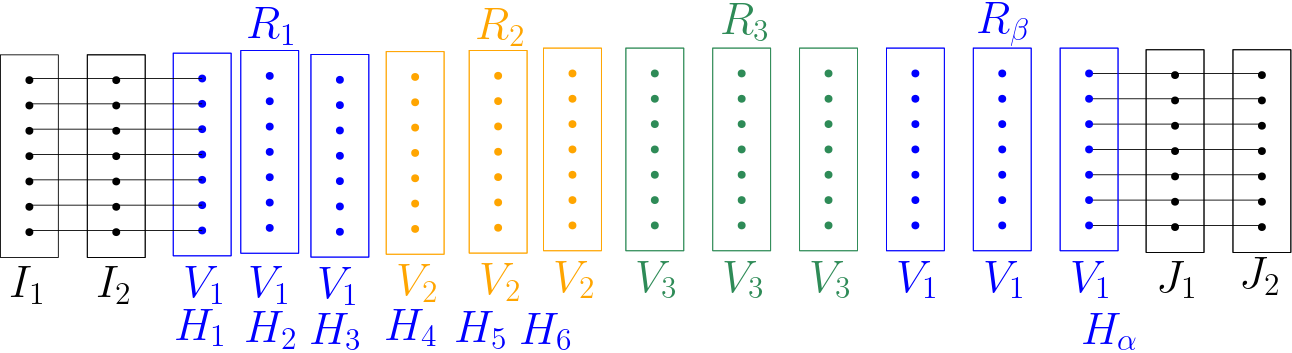}
    \caption{Example of a graph $\Gamma(2, H, 2) = \mathcal{I}^2 \oplus H \oplus \mathcal{J}^{2}$, where $\mu(R_1) = \mu(H_1) = \mu(H_2) = \mu(H_3) = V_1$, $\mu(R_2) = V_2$, $\mu(R_3) = V_3$, and $\mu(R_\beta) = V_1$. Edges inside $H$ are omitted. }
    \label{fig:buffergadget}
\end{figure}

If we denote by $I_i$ the sets of $\mathcal{I}^p$ and $J_i$ the sets of $\mathcal{J}^q$,
for $i \in [p + \alpha + q]$, we call $L_i$ the $i$-th layer of $\Gamma(p, H, q)$, where $L_i = I_i$ when $i \leq p$, $L_i = H_{i - q}$ when $p < i \leq p + \alpha$, and $L_i = I_{i - (q + \alpha)}$ when $i > q + \alpha$.

Let us now define collapses which are an important tool in our reduction. Let $i \le \kappa$ and $v_j \in V_i$.
We let $H(v^i_j)$ denote the induced subgraph obtained from $H$ by deleting in $H_{i'}$, for every $i'$ such that $\nu(i')=i$, all the vertices of $H_{i'}$ but the vertex $h$ such that $\pi(h)=v^i_j$. That is, we restrict all the layers of $H$ corresponding to $V_{\mu(i)}$ to a single vertex (the same vertex).
By abuse of notation, We will sometimes use $H(h^i_j)$, for some vertex $h_j^i \in H_i$, to denote the graph $H(v^{a}_j)$ where $a=\nu(i)$.  We say that $H(h_j^i)$ is a \emph{collapse of $H$ on $h_j^i$}, or, equivalently, collapsing $H$ on $h_j^i$ results in $H(h_j^i)$ (see Figure~\ref{fig:collapsegadget}). 
Note that in the graph $H(v_j^i)$ we are indeed forced to be properly colored for the $i$th set since we have deleted all the other vertices in $V_i$.

\begin{figure}[ht]
    \centering
    \includegraphics[scale=0.3]{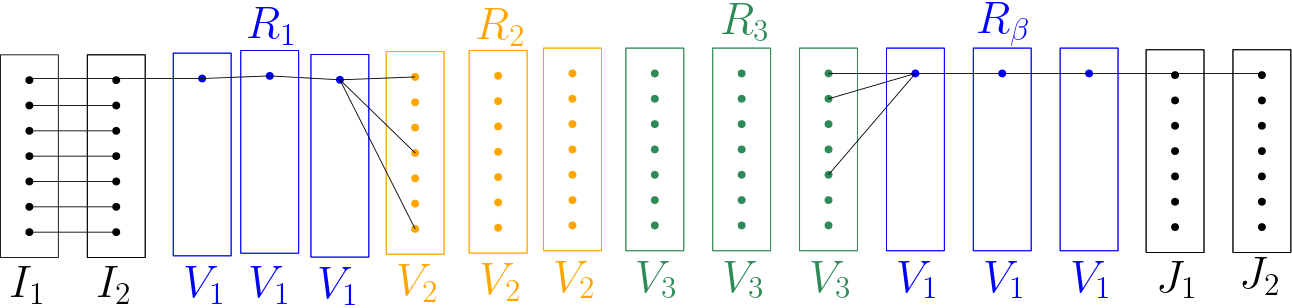}
    \caption{Example of a graph $\Gamma(2, H(h^1_1), 2)$ obtained after collapsing $H$ on $h^1_1$.}
    \label{fig:collapsegadget}
\end{figure}



Now, for every $i \le \kappa,j \le n$, we define $\Gamma_{i,j}(p,q)$ as $\Gamma(p, H(h_j^i), q)$  (it can be interpreted as the graph $H_j^i$ where we added some empty space before and after it). Finally, we let $\Gamma^i(p,q)$ denote the union of the $n$ graphs $\Gamma_{i,j}(p,q)$. We write $\Gamma_{i,j} = \Gamma_{i,j}(p,q)$ whenever $p$ and $q$ are clear from context. Note that all the $\Gamma_{i,j}$ being disjoint, if we have a path fully included in one of the $\Gamma_{i,j}(p,q)$ at some point, then all the selected vertices in sets mapping to $i$ by $\mu$ are the same. That is, $\Gamma^i(p,q)$ will allow us to verify that for any $H_j,H_j'$ in the selection gadget such that $\mu(j) = \mu(j') = i$ we always pick vertices $a \in H_j$, $b \in H_{j'}$ such that $\pi(a) = \pi(b)$. 

\subparagraph{Construction.}
We are now ready to describe the construction of the instance $(G',s,t,P,Q)$ of SPR. We begin by considering the token jumping model and discuss the changes required for sliding later. 

We start from an empty graph $G'$ and add two new vertices $s$ and $t$. We let $q = 2\kappa^2$ and $\delta = 2q + \alpha = 10\kappa^2$ (recall that  $\alpha = 6\kappa^2$). We add two internally vertex-disjoint $st$-paths $P$ and $Q$ consisting of $\delta$ internal vertices each. 

The next step consists of adding $\Gamma_\star = \Gamma(q, H, q)$ to $G'$ and connecting $s$ to every vertex in $I_1$ and $t$ to every vertex in $J_q$. Moreover, we let the $i$th internal vertex of $P$, $i \geq 2$, be adjacent to every vertex in layer $i-1$ of $\Gamma_\star$. 
We call $\Gamma_\star$ the \emph{selection gadget}. The rest of the gadgets will be verification and boundary gadgets that allow us to guarantee that properties similar to those in Observation~\ref{obs:goodpath} will hold. 

We then create a graph $\Gamma^{1}(q-1,q+1)$ denoted by $\Gamma^1$ which will be the \emph{verification gadget for $i = 1$}. 
We deal with the graphs of $\Gamma^1$ first 
(and slightly differently than the rest) as they require special attention given that they exist at the ``boundary'' of our construction. Notice that, in $G'$, all the graphs in $\Gamma^1$ are ``shifted one position to the left with respect to $\Gamma_\star$'' (in the sense that the number of independent sets at the left has reduced by one), see Figure~\ref{fig:jumpingreduction} for an illustration. In particular, the graph $H$ of each $\Gamma_{1,j}$, $j \in [n]$, starts (or appears) one layer before the graph $H$ in $\Gamma_\star$. We now describe the edges between $\Gamma_\star$ and any $\Gamma_{1,j}$ (in $\Gamma^1$). Let $L$ denote some layer of $\Gamma_{1,j}$ (ignoring the last layer) and let $L'$ be the layer after $L$ in $\Gamma_\star$. If  $L$ and $L'$ correspond to independent sets (not sets of $H$) they are connected by the natural matching. Otherwise, we have two cases:

\begin{itemize}
    \item If layer $L$ of $\Gamma_{1,j}$ corresponds to a set $H_p$ with $\mu(p) = 1$ then we deleted all vertices of $L$ except for $h_j^{p}$ (collapse). We connect $h_j^{p}$ to its image in $L'$, which must exists since layer $L'$ of $\Gamma_\star$ corresponds to a set $H_{p'}$ with $\mu(p) = \mu(p') = 1$. 
    \item Otherwise, we have the same number of vertices in $L$ and $L'$ and we add a matching between the pairs of vertices having the same image in $G$. 
\end{itemize}

We now add a boundary gadget that will separate all the verification gadgets and allow us to simplify some of the arguments. Picturing the graph being constructed from top to bottom with $P$ and $Q$ encircling all of the graph, we assume that $\Gamma_{i,j}$ is drawn before $\Gamma_{i,j+1}$. Similarly, we only insert  $\Gamma_{i + 1,j}$ after inserting all graphs of $\Gamma^i$ (see again Figure~\ref{fig:jumpingreduction} for an illustration). After $\Gamma_{1,n}$ is inserted, we insert another graph (connecting $s$ and $t$) that we denote by $\Gamma_{1,\star} = \Gamma(q - 2, H, q + 2)$ which is called the \emph{boundary gadget of $\Gamma^1$}.
Note that $\Gamma_{1,\star}$ is again shifted one position to the left compared to all the graphs in $\Gamma^1$. We add edges between layers of $\Gamma_{1,\star}$ and layers of $\Gamma_{1,j}$, for each $j \in [n]$. Like before, we let $L$ denote some layer of $\Gamma_{1,\star}$ (ignoring the last layer) and let $L'$ be the layer after $L$ in $\Gamma_{1,j}$. If $L$ and $L'$ correspond to independent sets (not sets of an $H$) then we connect them via a matching in the natural way. Otherwise, we have again two cases: 

\begin{itemize}
    \item $|L| = n$, $|L'| = 1$, and we connect by an edge the unique vertex of $L'$ to its image in $L$; or 
    \item $|L| = |L'| = n$ (by construction) and we connect the two layers by a matching. 
\end{itemize}

Informally speaking, we collapse the set $V_1$ into $n$ sets in all the layers $i'$ satisfying $\nu(i')=i$ between $\Gamma_\star$ and $\Gamma_1$ which replaces a single connected graph by $n$ of them. And then, we group them back in the gadget $\Gamma_{1,\star}$. This collapse will permit to ensure that we select  the same vertex everywhere and the grouping back phase will permit to avoid a combinatorial explosion (i.e. will allow to check the coherences of the different sets in $V_i$ independently).

We can now complete the construction as follows. For $i \in [\kappa - 1]$, after $\Gamma_{i,\star}$ is inserted we proceed just like before by assuming that $\Gamma_{i,\star}$ now takes the role of the \ selection gadget $\Gamma_\star$. Formally, for $i \in [\kappa - 1]$ and $j \in [n]$ (processing in increasing order), we create a graph $\Gamma_{i + 1,j}$, where $\Gamma_{i + 1,j} = \Gamma(q - (2i + 1), H(h_j^{i+1}), q + (2i + 1))$. We connect $s$ to all the first-layer vertices and $t$ to all the last-layer vertices in the obvious way. Let $\Gamma^{i + 1}$ denote the collection of the $n$ graphs of the form $\Gamma_{i + 1,j}$. We add edges between $\Gamma_{i, \star}$ and graphs in $\Gamma^{i + 1}$ just like before. Similarly, we then add a new graph $\Gamma_{i + 1,\star}$ and proceed as described until we reach $\Gamma_{\kappa,\star}$. We connect all the vertices of a layer of $\Gamma_{\kappa,\star}$ to the vertex of $Q$ on the preceding layer (see Figure~\ref{fig:jumpingreduction}). This completes the construction of the SPR instance $(G',s,t,P,Q)$\footnote{We note that most of the buffer space ``to the right'' of the construction is not needed but was added to favor a symmetric construction.}. Note that $|V(P)| = |V(Q)| = 10\kappa^2 + 2$. 

\subparagraph{Safeness of the reduction.}
Before we dive into the technical details of the proof, let us give some high-level intuition. Simply put, the purpose of every set of graphs $\Gamma^{i}$, $i \in [\kappa]$, is to verify that all the sets/layers of $\Gamma_\star$ mapping to the same $V_i$ use the same vertex of $V_i$. The trickier part of the proof is in showing that tokens are ``well-behaved''. 

\begin{figure}[ht]
    \centering
    \includegraphics[scale=0.25]{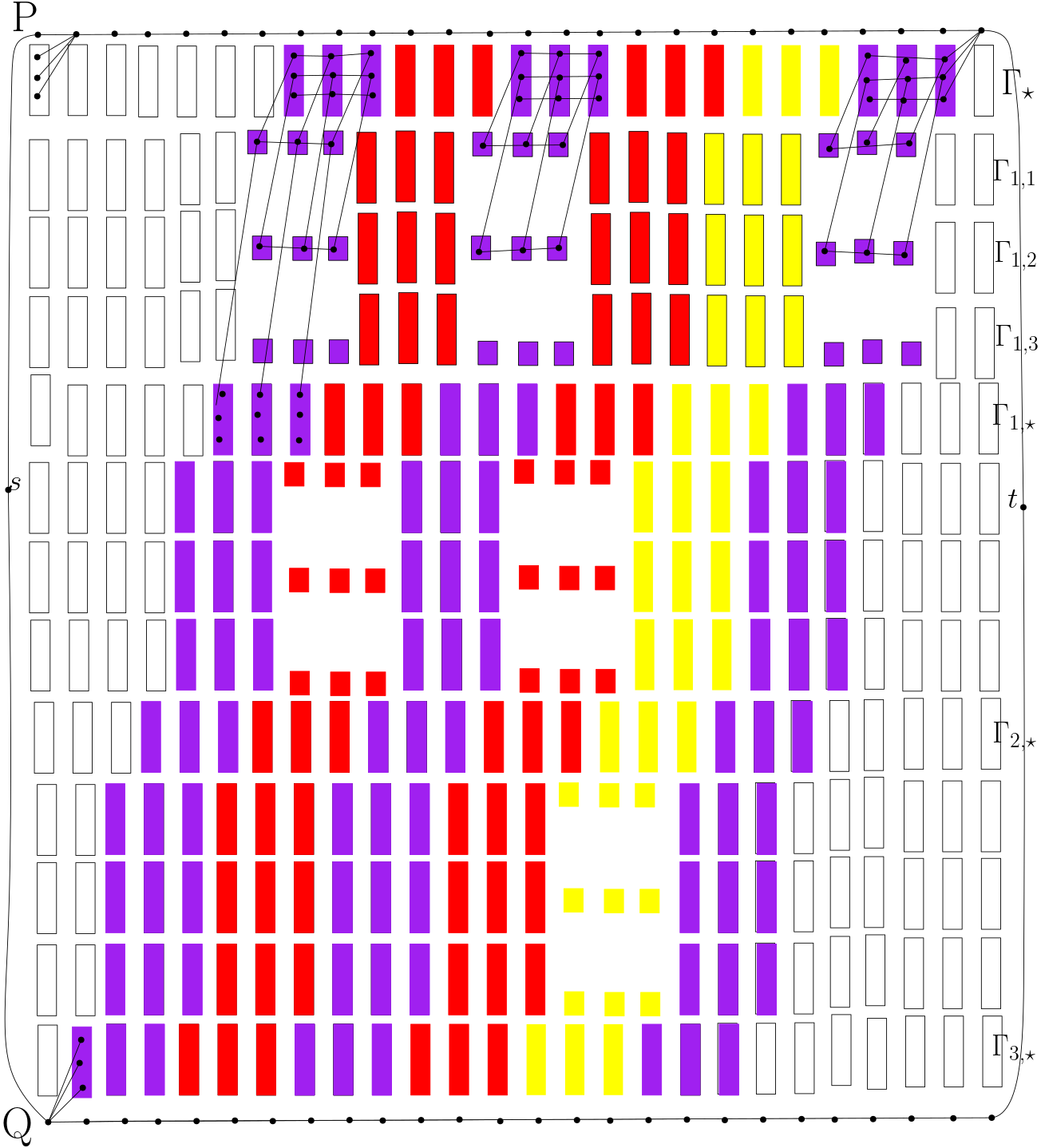}
    \caption{An example of our reduction in the case of token jumping.}
    \label{fig:jumpingreduction}
\end{figure}

Let us start by proving the easier direction. We assume, without loss of generality, that all of our gadgets $H$ start with a copy of $V_1$ and end with a copy of $V_\kappa$. Moreover no two consecutive groups of any $H$ map to the same $V_i$.  

\begin{lemma}\label{lem:hard_easydir}
If $(G,\kappa)$ is a yes-instance of \textsc{(Regular) Multicolored Clique} then there exists a reconfiguration sequence from $P$ to $Q$ whose length is $20(\kappa^3 + \kappa^2)$. 
\end{lemma}

\begin{proof}
Let $\{v_{j_1}^1, v_{j_2}^2, \ldots, v_{j_i}^i, \ldots, v_{j_\kappa}^\kappa\}$ denote the vertices of a multicolored clique in $G$. Let us exhibit a reconfiguration sequence from $P$ to $Q$. To do so, let us first give a reconfiguration sequence from $P$ to a path that contains vertices in $\Gamma_\star$ as follows:
\begin{itemize}
    \item We move one  by one the tokens of $P$ to $\Gamma_\star$ by increasing distance to $s$ (in ascending order). 
    \item For every layer $i \le q$, we jump (in order) the token at layer $i \geq 1$ in $P$ to vertex $v_{j_1}$ in the $i$th layer of $\Gamma_\star = \Gamma(q, H, q)$ as long as $i \leq q + 1$ (as $H_1$ maps to $V_1$ by assumption). In other words, we map all the vertices at the beginning of the path to the copy of vertex $v_{j_1}^1$.
    \item Then, for any layer $q + 1 < i \leq q + 1 + \alpha$, we jump the token at layer $i$ of $P$ to vertex $h_{j_{\mu(i)}}^{\mu(i)}$ of $\Gamma_\star$. 
    \item For every $i>q+1+\alpha$, we jump the $i$th vertex of $P$ to vertex $v_{j_\kappa}$ (since we assume that $H$ ends with a set that maps to $V_\kappa$). 
\end{itemize}
The fact that we maintain an $st$-path after every token jump follows from Observation~\ref{obs:goodpath} combined with the fact that vertices of $P$ are connected to all vertices of the preceding layer of $\Gamma_\star$. 

Once we have reached a properly colored $st$-path $P_1$ fully contained in $\Gamma_\star$ (in exactly $10\kappa^2$ steps), we can use a similar  strategy to reach a properly colored $st$-path $P_2$ fully contained in $\Gamma_{1,j_1}$. More formally, we move by increasing order all the tokens of $P_1$ in such a way the $i$-th vertex of $P_2$ is a the copy of the $(i+1)$-th vertex of $P_1$. Note that it is well-defined since, for every $i$ such that $\mu(i)=1$, the vertex $h_{j_1}$ belongs to $P_2$. Observe that during that transformation the vertices ``shift one layer to the left''. We then use a similar transformation to transform $P_2$ into a path $P_3$ fully contained in $\Gamma_{1,\star}$. We use $20\kappa^2$ steps from $P_1$ to $P_3$. 

We repeat this procedure for every $2 \le i \le \kappa$ to transform the path in $\Gamma_{i-1,\star}$ into a path in $\Gamma_{i,\star}$ in $20\kappa^2$ jumps. Then we need an extra $10\kappa^2$ steps to go from $\Gamma_{\kappa,\star}$ to $Q$ (using the converse of the transformation from $P$ to $\Gamma_\star$). 
Hence, the length of the reconfiguration sequence is exactly $20(\kappa^3 + \kappa^2)$. 
\end{proof}

In order to prove the other direction, we first establish some useful properties of our construction. We let $\Gamma_{\star} = \Gamma_{0,0}$ and $\Gamma_{i,\star} = \Gamma_{i,n+1}$. We say $\Gamma_{i,j}$ \emph{comes before or above} $\Gamma_{i',j'}$ whenever $i < i'$ or $i = i'$ and $j < j'$ (we also assume that $P$ appears first and $Q$ appears last, i.e., $P = \Gamma_{-1,-1}$ and $Q = \Gamma_{n+1,n+1}$). 
We say that two consecutive internal vertices $v_p$ and $v_{p+1}$ of an $st$-path $P$ are \emph{siblings} if they belong to the same graph $\Gamma_{i,j}$ (that is they belong to the same row in the representation of Figure~\ref{fig:jumpingreduction}). Otherwise, we say $v_p$ is \emph{above} (or \emph{below}) $v_{p+1}$ if the graph of $v_p$ is above (below) that of $v_{p+1}$ (that is $v_p$ is in the row above or below $v_{p+1}$ in the representation of Figure~\ref{fig:jumpingreduction}). 

\begin{lemma}\label{lem:props}
Let $P$ be a shortest path from $s$ to $t$ in $G'$. Let $v_p$ denote the $p$th internal vertex of $P$. Then: 
\begin{itemize}
    \item For every $p$, $v_p$ is a vertex of the $p$th layer of $G'$.
    \item For every two consecutive internal vertices of $P$, $v_p$ and $v_{p+1}$, either $v_p$ and $v_{p+1}$ are siblings or $v_p$ is below $v_{p+1}$. 
    \item For every $p$, if $v_p$ belongs to $\Gamma_{i,j}$ then no vertex $v_{p'}$ with $p' \ge p$ is below $v_p$. 
    \item For every $p$, if $v_p$ belongs to $\Gamma_{i,j}$ then $v_{p - 1}$ is either in $\Gamma_{i,j}$ or $\Gamma_{i,n + 1}$ and $v_{p + 1}$ is either in $\Gamma_{i,j}$ or $\Gamma_{i - 1, n + 1}$. 
\end{itemize}
\end{lemma}

\begin{proof}
The first point is trivial. 
The second point follows from the fact that there are no edges connecting a vertex of the $(p+1)$-th layer of $\Gamma_{i,j}$ to the $p$th layer of $\Gamma_{i',j'}$ with $i'<i$. The third point is a simple consequence of the second. The fourth point follows from the construction of the graph.
\end{proof}

Our next result states that the reconfiguration sequence described in Lemma~\ref{lem:hard_easydir} is the best possible. We say that a token $\tau$  makes a \emph{forward move} (\emph{backward move}) whenever the move decreases (increases) the number of moves still required by $\tau$ to reach $Q$. When a move causes neither a decrease nor an increase we say that token $\tau$ makes a \emph{local move}. We will prove that, if it exists, there is a transformation that does not make any local or backward move. That permits to control the shape of a shortest reconfiguration sequence and then ensures that this reconfiguration sequence can be interpreted as a multicolored clique of $G$.

\begin{lemma}\label{lem:hard_length}
Any reconfiguration sequence from $P$ to $Q$ requires at least $20(\kappa^3 + \kappa^2)$ token moves. Moreover, if there exists a reconfiguration sequence from $P$ to $Q$ then there exists one of length exactly $20(\kappa^3 + \kappa^2)$.
\end{lemma}

\begin{proof}
Recall that the number of internal vertices in any $st$-path in $G'$ is exactly $10\kappa^2$. Hence, to establish the lower bound of $20(\kappa^3 + \kappa^2)$, it suffices to show that every token on an internal vertex of $P$ has to make at least $2\kappa + 2$ forward moves. 

In order to go from the $i$-th vertex of $P$ to the $i$-th vertex of $Q$, one has to pass through a vertex of $\Gamma_{j,\star}$ for every $j$ in increasing order (since the vertices of $\Gamma_{j,\star}$ in the $i$-th layer separate $\Gamma_{j',r}$ with $j'\le j$ from $\Gamma_{j',r}$ with $j'>j$). As vertices of $\Gamma_{j,\star}$ are anti-complete to $\Gamma_{j',\star}$, at least $2\kappa$ moves are needed. Since we also need to move from $P$ to $\Gamma_*$ and from $\Gamma_{\kappa,\star}$ to $Q$, we get the desired bound.

So to conclude we have to prove that all the shortest reconfiguration sequences have length $20(\kappa^3 + \kappa^2) + 1$. Assume by contradiction that $P$ can be transformed to $Q$ and that every reconfiguration sequence from $P$ to $Q$ has length at least $20(\kappa^3 + \kappa^2) + 1$. Let $\sigma$ denote such a shortest reconfiguration sequence from $P$ to $Q$. It follows that at least one token must make at least $2\kappa + 3$ moves in $\sigma$. This implies that the token must make at least one local or backward move; a token on $P$ reaches $Q$ after a sequence of exactly $2\kappa + 2$ forward moves. We claim that the existence of a local or backward move in $\sigma$ contradicts the assumption that $\sigma$ is a shortest reconfiguration sequence. 

Consider the first token $\tau$ that makes a non-forward move in $\sigma$. Assume that $\tau$ makes a local move. By Lemma~\ref{lem:props}, when $\tau$ jumps from $v_p$ to $v'_p$ (in the same layer of the same row), it must be the case that $v_{p-1}$ is either a sibling of $v_p$ or below $v_p$ and $v_{p + 1}$ is either a sibling of $v_p$ or above $v_p$. If one of $v_{p - 1}$ or $v_{p + 1}$ is not a sibling of $v_p$ then no local move is possible; since only matching edges exist between non-siblings\footnote{Local moves are possible when $v_p$ is in $\Gamma_\star$ and $v_{p+1}$ is in $P$ but similar arguments imply that we can delete all local moves that happen while $v_{p+1}$ is in $P$ and only jump $\tau$ from $P$ to its final position in $\Gamma_\star$ immediately before the first move of the token on $v_{p+1}$. The same is true when $v_{p}$ is in $Q$ and $v_{p+1}$ is in $\Gamma_{\kappa,\star}$.}. Hence, both $v_{p - 1}$ and $v_{p + 1}$ must be sibling of $v_p$. Note that either $v_{p - 1}$ and $v_p$ or $v_{p}$ and $v_{p+1}$ must belong to the same group. Assume, without loss of generality, that $v_{p - 1}$ and $v_p$ belong to the same group. This implies that $\pi(v_{p - 1}) = \pi(v_p)$; since only matching edges are added within a group. Therefore, as $\pi(v'_p) \neq \pi(v_p) = \pi(v_{p - 1})$, $v'_p$ is not adjacent to $v_{p-1}$ and the local move is impossible, a contradiction. 

So we can assume that the first non-forward move in $\sigma$ is a backward move of token $\tau$. We denote the resulting path by $P_b$. We let $P_f$ denote the path resulting from the last forward move of $\tau$ prior to the backward move. We modify $\sigma$ to $\sigma'$ by deleting both the first backward move and the closest preceding forward move of $\tau$. We claim that $\sigma'$ is a valid reconfiguration sequence from $P$ to $Q$ that is shorter than $\sigma$, obtaining the required contradiction. 

To see why $\sigma'$ is valid, we consider the vertices occupied by $\tau$ (and its two neighbors) between $P_f$ and $P_b$. The last forward move of token $\tau$ jumps forward from some vertex $x_p$ above $z_p$ to $z_p$. Let $x_{p-1}$ and $x_{p + 1}$ denote the two neighbors of $x_p$ and $z_p$ in the path preceding $P_f$ and in $P_f$. 
The first backward jump of $\tau$ jumps from vertex $z_p$ to vertex $y_p$ above $z_p$. 
Let $y_{p-1}$ and $y_{p + 1}$ denote the two neighbors of $y_p$ and $z_p$ in the path preceding $P_b$ and in $P_b$. 
By Lemma~\ref{lem:props}, 
it must be the case that $x_{p-1}$ is a sibling of $z_p$ and below $x_p$ and $x_{p + 1}$ is a sibling of $x_p$ and above $z_p$. 
Also by Lemma~\ref{lem:props}, 
it must be the case that $y_{p-1}$ is a sibling of $z_p$ and below $y_p$ and $y_{p + 1}$ is a sibling of $y_p$ and above $z_p$. 
By construction, we have $\pi(x_{p-1}) = \pi(x_{p})$, $\pi(z_p) = \pi(x_{p + 1})$,
$\pi(y_{p-1}) = \pi(y_{p})$, and $\pi(z_p) = \pi(y_{p + 1})$. 

The token on $x_{p-1}$ cannot move between $P_f$ and $P_b$ as any forward move of said token will require a backward move prior to $P_b$, contradicting our choice of $\tau$. Consequently, we have $x_{p-1} = y_{p - 1}$. For the same reason, the token on $x_{p+1}$ cannot move between $P_f$ and $P_b$, implying  $x_{p+1} = y_{p+1}$. Putting it all together, we have $x_{p-1} = y_{p-1}$, $x_{p + 1} = y_{p+1}$, and $x_{p} = y_{p}$. Said differently, the token preceding $\tau$ and the token succeeding $\tau$ do not move between $P_f$ and $P_b$. Combined with the fact that $x_{p} = y_{p}$, we know that if we execute all moves of $\sigma$ prior to $P_b$ excluding the moves of $P_f$ and $P_b$ then we still reach $P_b$. However, this contradicts the fact that $\sigma$ is a shortest reconfiguration sequence. Hence, any shortest reconfiguration sequence from $P$ to $Q$ consists of exactly $20(\kappa^3 + \kappa^2)$ forward moves, as needed. 
\end{proof}

Given Lemma~\ref{lem:props} and Lemma~\ref{lem:hard_length}, it is easy to see that a shortest reconfiguration from $P$ to $Q$ in $G'$ must be \emph{monotone}, i.e., tokens always move towards $Q$ and every path in the reconfiguration sequence consists of a sequence of vertices (ordered from $s$ to $t$) whose distance from $Q$ monotonically increases. The last crucial brick in our proof requires a few  more definitions. We use $H$-layer to denote a layer in some $\Gamma_{i,j}$ that belongs to $H$. Moreover, when $\Gamma_{i,j}$ is clear from context,  we let $\mu^{-1}(i) = \{H_{j_1}, H_{j_2}, \ldots\}$ denote the set of all $H$-layers in $\Gamma_{i,j}$ that map to $V_i$, $i \in [\kappa]$. For a reconfiguration sequence $\sigma$, we let $\Gamma(\sigma)$ denote the set of all graphs $\Gamma$ that are touched by $\sigma$, i.e., a graph is touched by $\sigma$ if it contains a touched vertex and a vertex is touched by $\sigma$ if it ever receives a token. Note that $\Gamma(\sigma)$ contains the selection gadget, the boundary gadgets, and at least one graph from each verification gadget.  

\begin{lemma}\label{lem:clique-path}
Assume that there exists a reconfiguration sequence $\sigma$ from $P$ to $Q$ in $G'$. For $i \in [\kappa]$, let $\mu^{-1}(i) = \{H_{j_1}, H_{j_2}, \ldots\}$ denote the $H$-layers in $\Gamma_\star$ that map to $V_i$. Then: 

\begin{itemize}
    \item For every two consecutive sets $H_j$ and $H_{j+1}$ in $\Gamma \in \Gamma(\sigma)$ there exists at least one $st$-path $P'$ in the sequence $\sigma$ such that $P'$ contains one vertex in  $H_j$ and one vertex in $H_{j+1}$. 
    \item If $\sigma$ is a shortest reconfiguration sequence then the intersection of $\bigcup_{P' \in \sigma}{V(P')}$ with $\bigcup_{H_j \in \mu^{-1}(i)}{V(H_j)}$ includes only vertices that map to the same vertex of $V_i$. In other words, for any two vertices $w$ and $w'$ in $W = \bigcup_{P' \in \sigma}{V(P')} \cap \bigcup_{H_j \in \mu^{-1}(i)}{V(H_j)}$, we have $\pi(w) = \pi(w')$.
\end{itemize}
\end{lemma}

\begin{proof}
The fact that for every two consecutive sets in $\Gamma \in \Gamma(\sigma)$ there exists at least one $st$-path $P'$ in $\sigma$ that intersects with both sets follows from Lemma~\ref{lem:props}; otherwise a token jumps  beyond $\Gamma$, which is impossible when $\Gamma \in \Gamma(\sigma)$. 

Assume that there exist two sets $\{H_{j_1}, H_{j_2}\} \in \mu^{-1}(i)$ such that $P'$ intersect with $H_{j_1}$ on vertex $h^{j_1}_{a}$ and $P''$ (possibly equal to $P'$) intersects with $H_{j_2}$ on vertex $h^{j_2}_{b}$, where $\pi(h^{j_1}_{a}) \neq \pi(h^{j_2}_{b})$ (recall that $\mu(j_1) = \mu(j_2) = i$). Assume, without loss of generality, that $j_1 < j_2$ and $P'$ appears before $P''$ in $\sigma$ (and $H_{j_1},H_{j_2}$ belong to different groups). When $\sigma$ is a shortest reconfiguration sequence, it follows that no token ever makes a backward or local jump, i.e., a jump that does not bring the token closer to its final position in $Q$ (Lemma~\ref{lem:hard_length}). Hence, since $P'$ includes a vertex of $H_{j_1}$ all of its remaining vertices after $H_{j_1}$ (all the vertices in layers larger than $j_1$) must belong to $\Gamma_\star$ or $P$. In fact, the aforementioned vertices can be partitioned into two vertex-disjoint paths with the first path contained in $\Gamma_\star$ followed by a (possibly empty) path contained in $P$. The same is true for $P''$ and $H_{j_2}$. 

Consider the first time the token preceding but closest to the token on $h^{j_1}_{a}$ jumps to some vertex $x$ in some layer of $\Gamma_{i,p}$, $p \in [n]$. Let us denote the resulting path by $P_x$. Similarly, consider the first time the token preceding but closest to the token on $h^{j_2}_{b}$ jumps to some vertex $y$ in some layer of $\Gamma_{i,p'}$, $p' \in [n]$. Let us denote the resulting path by $P_y$. By construction, each layer of $\Gamma_{i,p}$ (or $\Gamma_{i,p'}$) corresponding to $V_i$ contains a unique vertex and all those vertices map to the same vertex of $V_i$ (the $p$th vertex or the $p'$th vertex, respectively). All vertices of $P_x$ after $x$ must be either at the same level or above $x$ (Lemma~\ref{lem:props}). The same is true for $P_y$ and all vertices of $P_y$ after $y$. We now show that it is impossible for $P_x$ to reach $P_y$ without backward moves, which contradicts the assumption that $\sigma$ is shortest. Assume otherwise. When the token jumps to $y$ it must be the case that the token after $y$ belongs to some vertex $w$ which is one level above $y$ and such that $\pi(y) = \pi(w)$ (by construction only matching edges are added). By repeated application of the same argument we get that $\pi(y) = \pi(h^{j_2}_{b}) \neq \pi(h^{j_1}_{a}) = \pi(x)$. Consequently, we have $\Gamma_{i,p} \neq \Gamma_{i,p'}$ (or simply $p \neq p'$). Assume, without loss of generality, that $p < p'$. Before any token of $P_x$ can jump to a vertex in $\Gamma_{i,p'}$ it must be the case that all tokens (in $\Gamma_{i,p}$) have moved to either  $\Gamma_{i-1,\star}$ and above or to  $\Gamma_{i,\star}$ and below. The former case implies at least one backward jump and the latter case implies at least one forward jump followed by a backward jump. In both cases, we get the required contradiction. 
\end{proof}

We now have all the ingredients to finish the proof.

\begin{lemma}\label{lem:hard-dir}
If $(G',s,t,P,Q)$ is a yes-instance of \textsc{Shortest Path Reconfiguration} then $(G,\kappa)$ is a yes-instance of \textsc{(Regular) Multicolored Clique}.
\end{lemma}

\begin{proof}
Let $(G',s,t,P,Q)$ be a yes-instance and let $\sigma$ be a shortest reconfiguration sequence from $P$ to $Q$. For $i \in [\kappa]$ and $P' \in \sigma$, let $W_i = \bigcup_{P' \in \sigma}{V(P')} \cap \bigcup_{H_j \in \mu^{-1}(i)}{V(H_j)}$. Moreover, let $\pi(W_i) = \{\pi(w) \mid w \in W_i\}$. By Lemma~\ref{lem:clique-path}, we have $|\pi(W_i)| = 1$ and we denote the vertex by $v^i_{j_i}$. Consider the $\kappa$ vertices $\{v^1_{j_1}, \ldots, v^i_{j_i}, \ldots, v^\kappa_{j_\kappa}\}$. The fact that those vertices must form a multicolored clique in $G$ again follows from Lemma~\ref{lem:clique-path}; as every pair must appear consecutively in two $H$-layers of $\Gamma_\star$ and some path of $\sigma$ must intersect with both. This completes the proof. 
\end{proof}

\begin{figure}[ht]
    \centering
    \includegraphics[scale=0.4]{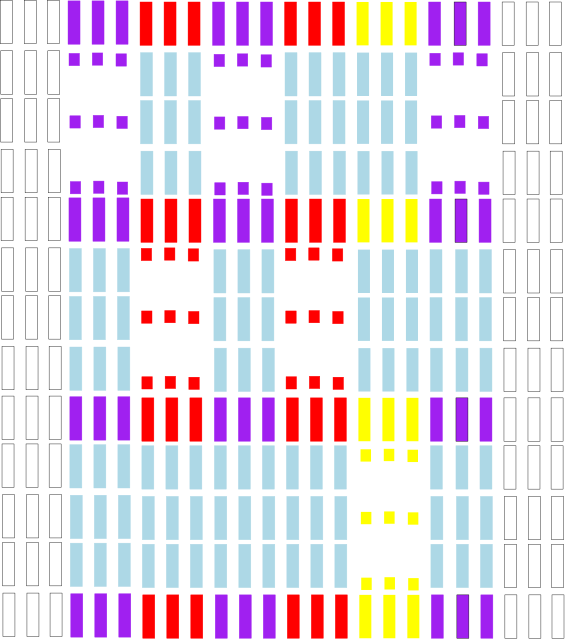}
    \caption{An example of our reduction in the case of token sliding.}
    \label{fig:slidingreduction}
\end{figure}

\begin{corollary}\label{cor:hard-models}
SPR is \textsf{W[1]}-hard parameterized by $k$ and SSPR is \textsf{W[1]}-hard parameterized by $k+\ell$ under the token jumping model.
\end{corollary}

\begin{proof}
The \textsf{W[1]}-hardness of SPR under the token jumping model follows from Lemmas~\ref{lem:hard_easydir} and~\ref{lem:hard-dir}. The \textsf{W[1]}-hardness of SSPR under the token jumping model follows by combining the aforementioned lemmas with Lemma~\ref{lem:hard_length}. 
\end{proof}

\begin{corollary}\label{thm:hardness_ts}
SPR is \textsf{W[1]}-hard parameterized by $k$ and SSPR is \textsf{W[1]}-hard parameterized by $k+\ell$ under the token sliding model.
\end{corollary}

\begin{proof}
We explain how to adapt the reduction for the token sliding model. We modify our construction in two ways. We ``align'' all the gadgets and add ``vertical matchings'' (see Figure~\ref{fig:slidingreduction}). Formally, all of our graphs $\Gamma_{i,j}$ will be either of the form $\Gamma(3, H, 3)$ or $\Gamma(3, H(h), 3)$. We add edges between sets appearing at consecutive layers in the natural way. 

We describe the edges between $\Gamma_{i-1,\star}$, $\Gamma_{i,j}$, and $\Gamma_{i,\star}$; there are no edges between the different graphs in $\Gamma^i$. Let $L$ and $L'$ denote two consecutive layers of $\Gamma_{i,j}$, let $M$ and $M'$ denote the same two layers (as $L$ and $L'$) in $\Gamma_{i-1,\star}$, and let $N$ and $N'$ denote the same two layers (as $L$ and $L'$) in $\Gamma_{i,\star}$. We add a matching from $M$ to $L$ and a matching from $L$ to $N$ (when $L$ consist of a single vertex we instead add a single edge to the corresponding image). We connect $L$ and $M'$ as well as $N$ and $L'$ the same way we connect $L$ and $L'$. That is, if $L$ and $L'$ belong to different groups, i.e., they map to different sets, we connect $L$ and $M'$, $N$ and $L'$, and $L$ and $L'$ using edges of $G$. If $L$ and $L'$ belong to the same group, we connect $L$ and $M'$, $N$ and $L'$, and $L$ and $L'$ using matchings. 

Observe that every vertex has at most one neighbor in the layer above it (when it exists) and at most one neighbor in the layer below it (when it exists). Combining this observation with the fact that we are in the token sliding model ensures that we are always selecting the same vertex in each layer. Hence, we can mimic the proof of Lemma~\ref{lem:clique-path} to conclude the proof.
\end{proof}
\subsection{Graphs of bounded degeneracy}
The goal of this section is to adapt the previous reduction to show that both problems remain hard on graphs of bounded degeneracy. 

\subsubsection{Token jumping}
We start with the case of token jumping. 
Note that the previous reduction offers two rather general (informal) properties:

\begin{enumerate}
    \item We can guarantee the selection of adjacent pairs of elements from different sets; and  
    \item we can guarantee that the same element is selected from a collection of sets.  
\end{enumerate}

Before diving into the details of proof, let us try to understand why we might have large degree or large degeneracy in our constructed graphs. If we look at a vertex $v$ in a set of $\Gamma_{i,j}$, then, by construction, it has no neighbor in the different graphs of $\Gamma^i$, no neighbor in sets just above or below it, no neighbor in sets below and to the right or above and to the left of it (Figure~\ref{fig:jumpingreduction}). In fact, in addition to neighbors in preceding and succeeding layers of $\Gamma_{i,j}$, $v$ has exactly one neighbor in one set below and to the left  and exactly one neighbor in one set above and to the right. Hence, in order to reduce the degeneracy one simply has to reduce the degeneracy in every graph $\Gamma_{i,j}$. To do so, we will reduce the degeneracy of the graph $G'$ by subdividing the edges of $G$ and adding sets to represent those edges. Recall that we start with an instance $(G,\kappa)$ of RMC where each vertex in $G$ has exactly $r$ neighbors in each of the $\kappa - 1$ other sets. We let $E_{i,j}$ denote the set of edges between vertices in $V_i$ and vertices in $V_j$. The idea is that in the construction of the previous section, we add a new set $V_{i,j}$ between every consecutive pair of sets $V_i,V_j$. In other words, we replace $E_{i,j}$ by a set $V_{i,j}$ containing one vertex $v_e$ for each edge $e \in E_{i,j}$. For $v_e \in V_{i,j}$, where $e = \{u,w\}$, we add the edges $\{v_e, u\}$ and $\{v_e, w\}$. 
We refine the partition of $V$ into sets $\{V_i \mid i \in [\kappa]\}$ and sets $\{V_{i,j} \mid (i,j) \in [\kappa] \times [\kappa]\}$. We let $m = |V_{i,j}|$. Observe that finding a multicolored clique now corresponds to finding the appropriate $\kappa$ vertices as well as the ${ \kappa \choose  2}$ edges connecting them. However, as long as we guarantee a consistent selection of actual graph vertices (from $V$), the consistency of edge-vertices will follow. 

\subparagraph{Construction of $G'$.}
We start with the construction of a gadget $H$ as before with $6\kappa^2$ sets of vertices 
$H_1, H_2, \ldots, H_{6\kappa^2}$ such that $|H_i| = n$ for each $i \in [6\kappa^2]$. We group every three consecutive sets into $2\kappa^2$ groups $R_1 = \{H_1, H_2, H_3\}$, $R_2 = \{H_4, H_5, H_6\}$, $R_3 = \{H_7, H_8, H_9\}$, $\ldots$, and $R_{2\kappa^2} = \{H_{6\kappa^2 - 2}, H_{6\kappa^2 - 1}, H_{6\kappa^2}\}$. We again assume, without loss of generality, that $H_1 = V_1$, $H_{6\kappa^2} = V_\kappa$, and no two consecutive groups of $H$ map to the same $V_i$.  
We then modify $H$ (without changing its name) by adding a set $H_{i,j}$ between every consecutive two groups $R_i = \{H_{h}, H_{h+1}, H_{h+2}\}$ followed by $R_j = \{H_{h+3}, H_{h+4}, H_{h+5}\}$. Every $H_{i,j}$ is a copy of $V_{\mu(i),\mu(j)}$ and every vertex $e_{a,b}$ with $a \in V_{\mu(i)}$ and $b \in V_{\mu(j)}$ is connected to exactly the copy of $a$ in $H_{h+2}$ and the copy of $b$ in $H_{h+3}$. 
We call $H_i$ a vertex layer of $H$ and $H_{i,j}$ an edge layer of $H$. We number the edge layers starting from $1$; $H$ now consists of $6\kappa^2$ vertex layers (divided into $2\kappa^2$ groups) and $2\kappa^2 - 1$ edge layers, i.e., we have a total of $8\kappa^2 - 1$ layers. 

We again maintain a mapping $\mu$ such that each $H_i$ corresponds to a copy of some $V_j$, for $i \in [6\kappa^2]$ and $j \in [\kappa]$. We also maintain a mapping $\pi$ that maps every vertex of $H_i$ to its corresponding vertex in $V_{\mu(i)}$. Given $H$ and a vertex $h_j^i \in H_i$, we let $H(h_j^i)$ denote the graph obtained from $H$ by deleting all vertices from each set $H_{i'}$, where $\mu(i') = \mu(i)$, except for $h_j^{i'}$.

We now describe the construction of the instance $(G',s,t,P,Q)$ of SPR (which is adapted from the one of the proof of Theorem~\ref{thm:hardness1}), where $G'$ has degeneracy four. We first add two new vertices $s$ and $t$ and two internally vertex-disjoint $st$-paths $P$ and $Q$ consisting of $24\kappa^2 - 1$ internal vertices each. The next step consists of adding $\Gamma_{0,\star} = \Gamma(8\kappa^2, H, 8\kappa^2)$ to $G'$ and connecting $s$ to every vertex in $I_1$ and $t$ to every vertex in $J_{8\kappa^2}$. Moreover, we let the $i$th internal vertex of $P$, $i \geq 2$, be adjacent to every vertex in layer $i-1$ of $\Gamma_{0,\star}$. 
After $\Gamma_{i,\star}$ is inserted we proceed as follows: 

\begin{itemize}
    \item For $i \in [\kappa - 1]$ and for every $j \in [n]$ (processing in increasing order), we insert a graph $\Gamma_{i + 1,j}$, where $$\Gamma_{i + 1,j} = \Gamma(8\kappa^2 - (2i + 1), H(h_j^{i+1}), 8\kappa^2 + (2i + 1)).$$ 
    \item Next, we insert a graph $\Gamma_{i + 1, \star}$, where $$\Gamma_{i + 1, \star} = \Gamma(8\kappa^2 - (2i + 2), H, 8\kappa^2 + (2i + 2)).$$ 
\end{itemize}

We connect $s$ to all the first-layer vertices and $t$ to all the last-layer vertices in the obvious way (that is we connect $s$ to the first independent set at the left of $\Gamma$ and $t$ to the last independent set at the right of $\Gamma$). Let $\Gamma^{i + 1}$ denote the collection of the $n$ graphs of the form $\Gamma_{i + 1,j}$. We add edges between $\Gamma_{i, \star}$ and graphs in $\Gamma^{i + 1}$ just like before. Similarly, we then add a new graph $\Gamma_{i + 1,\star}$ and proceed as described until we reach $\Gamma_{\kappa,\star}$. We connect all the vertices of a layer of $\Gamma_{\kappa,\star}$ to the vertex of $Q$ on the preceding layer.

\begin{lemma}\label{lem:4degenerate}
The graph $G'$ is $4$-degenerate. 
\end{lemma}

\begin{proof}
We start by deleting all the vertices in edge layers of graphs $\Gamma_{i,j}$, where $i \in [\kappa]$ and $j \neq \star$. Recall that each such vertex has a neighbor in the preceding layer of $\Gamma_{i,j}$, a neighbor in the succeeding layer of $\Gamma_{i,j}$, a neighbor in $\Gamma_{i,\star}$, and a neighbor in $\Gamma_{i+1,\star}$. After deleting all vertices in edge layers of $\Gamma_{i,j}$, the vertices of $\Gamma_{i,j}$ belonging to vertex layers become adjacent to at most four vertices;  a neighbor in $\Gamma_{i,\star}$, a neighbor in $\Gamma_{i+1,\star}$, a possible neighbor in $I_q \cup J_1$, and at most two neighbors within a group (have two neighbors inside a group implies no neighbors in $I_q \cup J_1$). We delete those vertices followed by deleting all remaining vertices of $\Gamma_{i,j}$. Consequently, we can then remove all graphs of the form $\Gamma_{i,\star}$, where $2 \leq i \leq \kappa - 1$, starting with the edge-layer vertices. This leaves $P$, $Q$, $\Gamma_{1,\star}$, and $\Gamma_{\kappa,\star}$. We proceed by deleting $\Gamma_{1,\star}$ and $\Gamma_{\kappa,\star}$ again starting with vertices in edge layers. 
We are now left with a cycle formed by $P$ and $Q$, which completes the proof.    
\end{proof}

\begin{corollary}
SPR is \textsf{W[1]}-hard parameterized by $k$ and SSPR is \textsf{W[1]}-hard parameterized by $k+\ell$ under the token jumping model, even when restricted to $4$-degenerate graphs. 
\end{corollary}

\begin{proof}
Hardness on $4$-degenerate graphs under the token jumping model follows from an easy adaptation of Corollary~\ref{cor:hard-models} and Lemma~\ref{lem:4degenerate}. 
\end{proof}

\subsubsection{Token sliding}
Let us explain how we can adjust the construction of the previous section in order to get the hardness proof under the token sliding model. We modify the reduction in a few ways. Informally, as in the proof of Corollary~\ref{thm:hardness_ts}, we will  ``align'' all the gadgets (see Figure~\ref{fig:slidingreduction}), add vertical edges, and replace some $H$ gadgets with $\triangle(H)$ gadgets to reduce the density in certain subgraphs~of~$G'$. 

We start with a few definitions. Given a vertex $v \in V_i$ and $i,i',i'' \in [\kappa]$ such that $i \neq i' \neq i''$, we let $\triangle(v, i, i', i'')$ denote a set of vertices containing one vertex for each pair $(u,w)$ such that $u \in V_{i'}$, $w\in V_{i''}$, and $\{v, u\}, \{v, w\}, \{u, w\} \in E(G)$. In other words, $\triangle(v, i, i', i'')$ contains a vertex for each triangle in $G$ that contains $v$, a vertex from $V_{i'}$,  and a vertex from $V_{i''}$. Given $v \in V_i$ and $i,i' \in [\kappa]$ such that $i \neq i'$, we let $\nabla(v, i, i')$ denote a set of vertices containing one vertex for each  edge in $G$ that is incident to $v$ and a vertex from $V_{i'}$.
We let $\triangle_{i, i', i''} = \bigcup_{v \in V_i}{\triangle(v, i, i', i'')}$ and $\nabla_{i, i'} = \bigcup_{v \in V_i}{\nabla(v, i, i')}$.

We introduce a new gadget $\triangle(H)$, similar to $H$, which consists of replacing each non-collapsed $H_i$, i.e., $|H_i| \geq 2$, with a copy of $\triangle_{\mu(i), \mu(i - 1), \mu(i + 1)}$ (instead of a copy of $V_{\mu(i)}$) if $\mu(i) \neq \mu(i - 1) \neq \mu(i + 1)$. If $\mu(i) = \mu(i - 1) \neq \mu(i + 1)$ we let $H_i = \nabla_{\mu(i), \mu(i + 1)}$,  
and if $\mu(i) = \mu(i + 1) \neq \mu(i - 1)$ we let $H_i = \nabla_{\mu(i), \mu(i - 1)}$. We set 
$H_1 = \nabla_{\mu(1), \mu(2)}$ and $H_{6\kappa^2} = \nabla_{\mu(6\kappa^2), \mu(6\kappa^2 - 1)}$. 
We update the mapping $\pi: H_i \rightarrow V_{\mu(i)}$ to map every vertex of $H_i$ to its corresponding vertex in $V_{\mu(i)}$. We construct a new mapping $\pi_\triangle: H_i \rightarrow V_{\mu(i - 1)} \times V_{\mu(i)} \times V_{\mu(i + 1)}$ to map every vertex of $H_i = \triangle_{\mu(i), \mu(i - 1), \mu(i + 1)}$ to its corresponding triangle. We construct a new mapping $\pi_\nabla: H_i \rightarrow \{V_{\mu(i - 1)} \times V_{\mu(i)}\} \cup \{V_{\mu(i + 1)} \times V_{\mu(i)}\}$ to map every vertex of $H_i = \nabla_{\mu(i), \mu(i - 1)}$ or $H_i = \nabla_{\mu(i), \mu(i + 1)}$ to its corresponding edge. We assume, without loss of generality, that $|H_i| = n$ even for $\triangle(H)$; all sets can be made of equal size by adding dummy vertices. 
Now, after $\Gamma_{i,\star}$ is inserted we proceed as follows: 

\begin{itemize}
    \item For $i \in [\kappa - 1]$ and for every $j \in [n]$ (processing in increasing order), we insert a graph $\Gamma_{i + 1,j}$, where $$\Gamma_{i + 1,j} = \Gamma(3, \triangle(H(h_j^{i+1})), 3).$$ 
    \item Next, we insert a graph $\Gamma_{i + 1, \star}$, where $$\Gamma_{i + 1, \star} = \Gamma(3, H, 3).$$ 
\end{itemize}

All of our graphs $\Gamma_{i,j}$ will be either of the form $\Gamma(3, H, 3)$ or $\Gamma(3, \triangle(H(h)), 3)$. We add edges between sets appearing in  consecutive layers of $\Gamma_{i,j}$ in the natural way. That is, whenever one of the two layers consists of an independent set and whenever both layers belong to the same group, we simply add a matching between the vertices. Otherwise, i.e., when we have $(H_i, H^j_{i,i+1})$, then we add an edge between a vertex $h \in H_i$ and a vertex $h_e \in H^j_{i,i+1}$, $\pi(h_e) = e = \{u,v\}$,  whenever one of the following is true:

\begin{itemize}
    \item $|H_i| = 1$ or $H_i = V_{\mu(i)}$ and $\pi(h) \in e$; 
    \item $H_i = \triangle_{\mu(i), \mu(i - 1), \mu(i + 1)}$ and $e \subseteq \pi_\triangle(h)$; 
    \item $H_i = \nabla_{\mu(i), \mu(i - 1)}$ and $\pi_\nabla(h) = e$; or
    \item $H_i = \nabla_{\mu(i), \mu(i + 1)}$ and $\pi_\nabla(h) = e$; 
\end{itemize}

When we have $(H^j_{i,i+1}, H_{i+1})$, then we add an edge between a vertex $h_e \in H^j_{i,i+1}$, $\pi(h_e) = e = \{u,v\}$, and a vertex $h \in H_{i+1}$ whenever one of the following is true:

\begin{itemize}
    \item $|H_{i+1}| = 1$ or $H_{i+1} = V_{\mu(i+1)}$ and $\pi(h) \in e$; 
    \item $H_{i+1} = \triangle_{\mu(i + 1), \mu(i), \mu(i + 2)}$ and $e \subseteq \pi_\triangle(h)$; 
    \item $H_{i+1} = \nabla_{\mu(i + 1), \mu(i)}$ and $\pi_\nabla(h) = e$; or
    \item $H_{i+1} = \nabla_{\mu(i + 1), \mu(i + 2)}$ and $\pi_\nabla(h) = e$; 
\end{itemize}

We describe the edges between $\Gamma_{i-1,\star}$, $\Gamma_{i,j}$, and $\Gamma_{i,\star}$; there are no edges between the different graphs in $\Gamma^i$. Let $L_p$ and $L_{p +1}$ denote two consecutive layers of $\Gamma_{i,j}$, let $M_{p}$ and $M_{p+1}$ denote the same two layers (as $L_p$ and $L_{p +1}$) in $\Gamma_{i-1,\star}$, and let $N_{p}$ and $N_{p+1}$ denote the same two layers (as $L_p$ and $L_{p +1}$) in $\Gamma_{i,\star}$. We add an edge from every vertex $u$ in $M_{p}$ to every vertex $v$ in $L_p$ whenever $\pi(v) \in \triangle(\pi(u), \mu(p), \mu(p-1), \mu(p+1))$ or $\pi(v) \in \nabla(\pi(u), \mu(p), \mu(p+1))$ or $\pi(v) \in \nabla(\pi(u), \mu(p), \mu(p-1))$. Similarly, we add an edge from every vertex $u$ in $N_{p}$ to every vertex $v$ in $L_p$ whenever $\pi(v) \in \triangle(\pi(u), \mu(p), \mu(p-1), \mu(p+1))$ or $\pi(v) \in \nabla(\pi(u), \mu(p), \mu(p+1))$ or $\pi(v) \in \nabla(\pi(u), \mu(p), \mu(p-1))$. When $L_p$ consist of a single vertex we instead add a single edge to the corresponding image. We connect $L_{p}$ and $M_{p+1}$ as well as $N_{p}$ and $L_{p+1}$ the same way we connect $L_{p}$ and $L_{p+1}$, i.e., as described above for adding edges between sets appearing in  consecutive layers of $\Gamma_{i,j}$. This completes the construction of the instance $(G',s,t,P,Q)$ of SPR.

\begin{lemma}\label{lem:6degenerate}
The graph $G'$ is $6$-degenerate. 
\end{lemma}

\begin{proof}
We repeatedly delete vertices of $G'$ of degree at most six, and, by a slight abuse of notation, keep using $G'$ to denote the resulting graph. We claim that this procedure will delete all vertices of $G'$. We prove the claim by contradiction, i.e., we assume that no vertex can be deleted from $G'$. Equivalently, we assume that $G'$ has minimum degree seven.   

Note that each vertex in an edge layer  of graph $\Gamma_{i,j}$, $i \in [\kappa]$ and $j \neq \star$,  has a neighbor in the preceding layer of $\Gamma_{i,j}$, a neighbor in the succeeding layer of $\Gamma_{i,j}$, two neighbors in $\Gamma_{i,\star}$, and two neighbors in $\Gamma_{i+1,\star}$. Hence, $G'$ cannot contain any edge-layer vertex in $\Gamma_{i,j}$, $i \in [\kappa]$ and $j \neq \star$. 
After deleting all vertices in edge layers of $\Gamma_{i,j}$, the vertices of $\Gamma_{i,j}$ belonging to vertex layers of size greater than one (not collapsed) become adjacent to at most six vertices; two neighbors in $\Gamma_{i,\star}$, two neighbors in $\Gamma_{i+1,\star}$, possibly a neighbor in $I_q$ or $J_1$, and at most two neighbors within the same group (having two neighbors inside the group implies no neighbors in $I_q$ or $J_1$). This follows from the fact that every such vertex either corresponds to an edge or a triangle. So, if some vertex of $\Gamma_{i,j}$ still exists in $G'$, it must be the case that the vertex is in a vertex layer of $H$ and has many neighbors in (the next layer of) $\Gamma_{i,\star}$ or (the previous layer of) $\Gamma_{i+1,\star}$. Such a vertex must belong to a collapsed layer and is the unique vertex in its layer. Let us consider the first $p$ (closest to $0$) such that $p$ is the index of the first $\Gamma_{p,j}$ that could not be completely deleted. By our choice of $p$, it must be the case that all graphs of the form $\Gamma_{p - 1,j}$, $j \in [n]$ no longer exist in $G'$. 
Consequently, the edge-layer vertices of $\Gamma_{p,\star}$ are adjacent to vertices at the same level or below them. In particular, an edge-layer vertex $v_e$ of $\Gamma_{p,\star}$ has only two neighbors in $\Gamma_{p,\star}$ and zero or more neighbors in $\Gamma_{p,j}$, $j \in [n]$. Hence, the vertex $v_e$ has degree more than six if and only if it has more than two neighbors contained in non-collapsed preceding vertex layers that come after $\Gamma_{p,\star}$. We know that this is impossible since all such vertices (in non-collapsed vertex layers of $\Gamma_{p,j}$) have been deleted already; as they have degree at most six after the edge-layers have been deleted. This contradicts our assumption that no vertex of $G'$ can be deleted, as needed to complete the proof.  
\end{proof}

\begin{corollary}
SPR is \textsf{W[1]}-hard parameterized by $k$ and SSPR is \textsf{W[1]}-hard parameterized by $k+\ell$ under the token sliding model, even when restricted to $6$-degenerate graphs.
\end{corollary}

\begin{proof}
Hardness on $6$-degenerate graphs under the token sliding model follows from Corollary~\ref{cor:hard-models} (for correctness) and Lemma~\ref{lem:6degenerate} (for bounding the degeneracy). 
\end{proof}

\section{\texorpdfstring{$\FPT$}{FPT} algorithms}
First, we observe that both SPR and SSPR are easily shown to be fixed-parameter tractable when parameterized by $k + \Delta(G)$, where $\Delta(G)$ denotes the maximum degree of $G$; by only retaining vertices that belong to some shortest $st$-path one can easily bound the size of the graph since the $i$-th layer, consisting of all the vertices at distance exactly $i$ from $s$, will contain at most $\Delta(G)^i$ vertices. In the remainder of this section, we investigate the complexity of the problem further (and for different parameters) in order to identify the boundary between tractability and intractability. 

We denote an instance of SSPR by a tuple $(G,s,t,P,Q,\ell)$, where $G$ denotes the input graph, $s$ and $t$ denote the starting and ending vertices of the shortest paths, respectively, and $P$ and $Q$ denote the source and target shortest paths, respectively. We use $k$ to denote $|P| = |Q|$ which is the number of edges along a shortest path from $s$ to $t$, i.e., the length of the path. 
\subsection{\texorpdfstring{$\FPT$}{FPT} for parameter \texorpdfstring{$\ell$}{l} on nowhere dense classes of graphs}

As a warm-up, let us first prove that the following holds:

\begin{lemma}\label{lem:nowhere_k+l}
SSPR is $\FPT$ parameterized by $k+\ell$ on nowhere-dense classes of graphs for both the sliding and the jumping models. 
\end{lemma}

\begin{proof}
The proof easily follows from the fact that FO-model checking is $\FPT$ on nowhere dense classes of graphs~\cite{DBLP:conf/stoc/GroheKS14}. Such an argument has already been used in various proofs for reconfiguration problems, see e.g.,~\cite{BousquetMNS22+}.

For every $i \le k$ and $j \le \ell$, let us create a variable $x_{i,j}$  that represents the $i$-th vertex of the path at the $j$-th step of the reconfiguration sequence. Let us prove that we can formulate the existence of a reconfiguration sequence of length $\ell$ between $P$ and $Q$ as a FO-formula on the set of variables $x_{i,j}$.

First we set $x_{i,1}=p_i$ where $p_i$ is the $i$-th vertex of the path $P$. Similarly $x_{i,\ell}=q_i$ where $q_i$ is the $i$-th vertex of the path $Q$.
We now need to ensure that at every step $j \le \ell$, the set of variables $x_{i,1},\ldots,x_{i,\ell}$ is a path of $G$, that is, for every $i \le k-1$ and every $j \le \ell$, $x_{i,j} x_{i+1,j}$ is an edge.

We further want to ensure that if one vertex is modified between the $j$-th path and the $(j+1)$-th path then all the other vertices are the same. That is, for every $i, i' \le k$ and $j \le \ell-1$, we have $(x_{i,j} \ne x_{i,j+1} \Rightarrow (x_{i',j}=x_{i',j+1})$.
If we want a reconfiguration sequence with the token sliding rule, we have to add the following constraint: for every $i \le k, j \le \ell-1$, $x_{i,j}=x_{i,j+1}$ or $x_{i,j}x_{i,j+1}$ is an edge.
Finally, we add the constraints  $x_{1,j}=s$ and $x_{k,j}=t$ for every $j \le \ell$. 

Let us denote by $\phi$ the resulting formula.
Let us prove that there exists a reconfiguration sequence from $P$ to $Q$ of length at most $\ell$ if and only if $\phi$ is satisfiable.

If there exists a reconfiguration sequence $P_1=P,\ldots,P_r=Q$ with $r \le \ell$ then we simply have to set $x_{i,j}$ to be the $i$-th vertex of $P_j$ and $x_{i,j'}=q_i$ for every $j' \ge r$ in order to satisfy all the constraints.

Conversely, assume that there exists an assignment of the variables that satisfies all the constraints. Let us denote by $P_j$ the set of ordered vertices $x_{i,j}$ for $1 \le i \le k$. Note that, by hypothesis, $P_j$ is an $st$-path for every $j$. Moreover, by definition $P_j$ and $P_{j+1}$ differ on at most one vertex and $P_1=P$ and $P_\ell=Q$. By removing consecutive paths that are the same we obtain a reconfiguration sequence from $P$ to $Q$, which completes the proof.
\end{proof}

Let us now generalize this result and prove that the following holds:

\begin{theorem}
SSPR is \texorpdfstring{$\FPT$}{FPT} parameterized by $\ell$ on nowhere dense classes of graphs for both the sliding and the jumping models.
\end{theorem}

\begin{proof}
The idea of the proof consists of proving that there exists an equivalent instance where the distance between $s$ and $t$ is bounded by a function of $\ell$. The conclusion then directly follows from Lemma~\ref{lem:nowhere_k+l}. To do so, we will prove that we can bound (by a function of $\ell$) the set of indices $i$ on which there is a relevant modification on the $i$-th vertex of the path at some step of the reconfiguration sequence. We will then prove that we can ``forget'' the vertices which are not in these positions by reducing the length of the shortest paths.

Let $(G,s,t,P,Q,\ell)$ be an instance of SSPR. Let us denote by $S$ the set of positions on which $P$ and $Q$ differ. Note that if $|S| > \ell$ then we can immediately return false since more than $\ell$ steps are needed to transform $P$ into $Q$. So we can assume that $|S| \le \ell$ in the rest of the proof.

\begin{claim}
If there is a reconfiguration sequence from $P$ to $Q$ of length at most $\ell$ then there is a reconfiguration sequence from $P$ to $Q$ that only modifies vertices whose indices are at distance at most $\ell$ from an index of $S$.
\end{claim}
\begin{proof}
Let $\mathcal{R}$ be a reconfiguration sequence from $P$ to $Q$ of length at most $\ell$. At each step, there is exactly one position where a vertex is modified. Let us denote by $R$ that set of positions where a vertex is modified. We have $|R|\le \ell$. 
A \emph{component} $R'$ of $R$ is a maximal subset of $R$ containing consecutive integers. Every component $R'$ has a minimum and a maximum value (that might be equal).  
We say that a component is \emph{important} if it contains a vertex of $S$ and \emph{useless} otherwise.

We claim that if there is a useless component $R'$, removing from $\mathcal{R}$ all the modifications at position $c$ for every $c \in R'$ leaves a reconfiguration sequence from $P$ to $Q$. Indeed, let us denote by $\mathcal{R}'$ the resulting reconfiguration sequence. First note that since $R'$ is a useless component, the final shortest path is still $Q$ (we cancel modifications on positions where $P$ and $Q$ were identical). Assume now, for a contradiction, that at some step of the reconfiguration sequence in $\mathcal{R}'$, the set of vertices $P_i$ is not a shortest $st$-path. Let us denote by $u,v$ the consecutive vertices of $P_i$ that are not adjacent. Since the path is only modified at positions of indices of $R'$, either the index of $u$ or $v$ is in $R'$. Moreover, both of them are not in $R'$ since by definition of $\mathcal{R}'$ all the vertices of indices in $R'$ remain the same all along the reconfiguration sequence and the initial set of vertices is indeed a path. So we can assume by symmetry that the position of $u$ is the index just before the minimum value of $R'$ and $v$ is the minimum value of $R'$. Since the vertex $v$ belongs to all the sets in the reconfiguration sequence $\mathcal{R}'$, it means that $u$ has been modified. But then $u$ should be added in the component $R'$ of $v$, a contradiction.  

Thus, if there is a reconfiguration sequence from $P$ to $Q$ of length $\ell$, there is one with no useless component. But the width of a component is at most $\ell$ since only $\ell$ vertices are modified in a reconfiguration sequence. So if there is a reconfiguration sequence, there is one that only moves tokens on vertices whose indices are at distance at most $\ell$ from an index of $S$, as claimed.
\end{proof}

Let $X(i,s)$ be the set of vertices at distance exactly $i$ from $s$ in $G$. Let $I_S$ be the set of indices at distance at most $\ell$ from an index of $S$. 
Note that $I_S$ has size at most $2\ell \cdot |S|$.
An empty interval for $I_S$ is an interval maximal by inclusion in $\{0,\ldots,d(s,t) \} \setminus I_S$. Note that $I_S$ has at most $|S|$ empty intervals. 
We create the graph $G'$ from $G$ as follows:
\begin{itemize}
\item For every $i \in I_S$, $G'$ contains all the vertices of $X(i,s)$.
\item For all the integers $i \notin I_S$ but at distance one from an integer of $I_S$, $G'$ contains the vertex at position $i$ in $P$ (and $Q$).
\item $G'$ contains $s$ and $t$.
\item There is an edge between $x$ and $y$ if $xy$ is an edge of $G$, or if $x,y$ are the unique two vertices of $G$ whose positions are in the same empty interval for $I_S$ \footnote{Informally, we link the vertex of $P$ just after an interval of $I_S$ with the vertex of $P$ just before the beginning of the next interval of $I_S$.}.
\end{itemize}
Let us denote by $P'$ and $Q'$ in $G'$ the set $P \cap V(G')$ and $Q \cap V(G')$. One can easily remark that $P'$ and $Q'$ are shortest $st$-paths in $G'$.

\begin{claim}
There is a reconfiguration sequence from $P$ to $Q$ in $G$ if and only if there is a reconfiguration sequence from $P'$ to $Q'$ in $G'$.
\end{claim}

\begin{proof}
The proof follows from the fact that we can assume that a transformation from $P$ to $Q$ of length at most $\ell$ in $G$ only modifies vertices whose indices are at distance at most $\ell$ from an index of $S$. All those vertices are in $G'$ and all the vertices of $G'$ that contain non-movable tokens are unique at their corresponding distance from $s$ (hence cannot move in $G'$). 
\end{proof}

One can remark that the distance between $s$ and $t$ in $G'$ is at most $4 \ell^2$. So by Lemma~\ref{lem:nowhere_k+l}, we can decide in \textsf{FPT}-time in $\ell$ if there is a reconfiguration sequence from $P'$ to $Q'$ in $G'$, which completes the proof. 
\end{proof}

\subsection{\texorpdfstring{$\FPT$}{FPT} for parameters cluster deletion number, treedepth, and modular width}

In this section we show that SPR and SSPR are fixed-parameter tractable when parameterized by $\textsf{cd}(G)$, $\textsf{td}(G)$, and $\textsf{mw}(G)$, where $\textsf{cd}(G)$, $\textsf{td}(G)$, and $\textsf{mw}(G)$ denote the cluster deletion number, treedepth, and modularwidth of $G$, respectively. For all three parameters, we present algorithms that compute reconfiguration sequences of minimum length; so we focus only on SPR. We conclude by proving that if we parameterize by the feedback vertex set number of the graph, $\textsf{fvs}(G)$, then we can assume that the length of shortest paths is also bounded by $\textsf{fvs}(G)$, i.e., $k \in \mathcal{O}(\textsf{fvs}(G))$. Despite this fact, the complexity of the problem remains open when parameterized by $\textsf{fvs}(G)$.

Let us first start by defining the aforementioned graph parameters and how they relate to each other (see Figure~\ref{fig:classes}). 
The \emph{cluster deletion number} of a graph $G$, denoted by $\textsf{cd}(G)$, is the minimum number of vertices whose deletion from $G$ results in a disjoint union of complete graphs, i.e., cliques (deleting a vertex also deletes all edges incident to it).

The \emph{treedepth} of a connected graph $G$, denoted by $\textsf{td}(G)$, is the minimum depth of a rooted tree $T$
on the vertices of $G$ such that each edge of $G$ connects an ancestor-descendant pair of $T$~\cite{DBLP:books/daglib/0030491}. Here, the depth of a tree is the length of the longest root-to-leaf path. For disconnected graphs, one can use a forest instead of a tree but we only consider connected graphs. 
Intuitively, just like graphs of large treewidth may be characterized by large grid minors, treedepth may be characterized by excluded paths, i.e., a graph has large treedepth if and only if it contains a long path~\cite{DBLP:journals/ejc/NesetrilM08}. 
As noted by Demaine et al.~\cite{DemaineEHJLUU19}, treedepth is a natural graph parameter to use for path reconfiguration since the maximum path-length of graphs of bounded treedepth is also bounded by treedepth. In fact, if a graph $G$ has maximum path-length $r$ then its treedepth can be at most $r$. In the other direction, a graph
with treedepth $r$ has maximum path-length at most $2^{r+1} - 2$~\cite{DemaineEHJLUU19}. It is also known that in graphs of treedepth $r$, the number of paths of a given length can be $\Theta(n^{2^r})$~\cite{DemaineEHJLUU19}.
Wrochna initiated the study of the parameterized complexity of reconfiguration problems on graphs of bounded treedepth and graphs of bounded bandwidth~\cite{Wrochna18}, showing, e.g., that \textsc{Shortest Path Reconfiguration} (in any of the token models) is \textsf{PSPACE}-complete on graphs of bounded bandwidth (hence pathwidth, treewidth, and cliquewidth).  

A \emph{vertex cover} of a graph is a set of vertices whose deletion leaves a graph with no edges, i.e., an independent set. The \emph{vertex cover number} of a graph $G$, denoted by $\textsf{vc}(G)$, is the size of a smallest vertex cover in $G$.
A \emph{feedback vertex set} of a graph is a set of vertices whose deletion leaves a graph without cycles, i.e., a forest. The \emph{feedback vertex set number} of a graph $G$, denoted by $\textsf{fvs}(G)$, is the size of a smallest feedback vertex set in $G$. To the best of our knowledge, no reconfiguration problems have been studied in graphs of bounded feedback vertex set number (with the exception of~\cite{DemaineEHJLUU19} that gives an \textsf{XP} algorithm). 

While most of our positive results in this section are relatively straightforward, the parameterization by $\textsf{fvs}(G)$ remains elusive. In fact, we have not succeeded in solving an even simpler problem. That is, we assume that we have a bounded-size modulator to a linear forest, i.e., every tree in the forest is a path (Figure~\ref{fig:hardinstances}). Nevertheless, we show that the feedback vertex set number is a natural parameter for shortest path reconfiguration problems since every positive instance must consist of shortest paths whose length is bounded by a function of the feedback vertex set number (otherwise we can determine immediately that we are dealing with a negative instance). It is also shown in~\cite{DemaineEHJLUU19} that in an $n$-vertex graph with feedback vertex set number $r$, the number of paths is at most $r!2^r\big ({n - r \choose 2} + n - r + 1\big )^{r + 1}$, which implies a trivial \textsf{XP} algorithm.

We also show that the  \textsc{Shortest Path Reconfiguration} problem is fixed-parameter tractable when parameterized by modularwidth~\cite{DBLP:conf/iwpec/GajarskyLO13}, where  the shortest path length is also bounded by the modularwidth. Combining our results with the $\PSPACE$-completeness result of Wrochna~\cite{Wrochna18}, the boundary separating the tractable instances from the intractable ones becomes much clearer (Figure~\ref{fig:classes}); with the obvious exception of the case of feedback vertex set number which remains open. 

In a graph $G$, a \emph{module} is a set of vertices $M \subseteq V(G)$ such that for all $u, v \in M$ and $w \in V (G) \setminus M$, if $\{u, w\} \in E(G)$ then $\{v, w\} \in E(G)$. In other words, a module is a set of vertices that have the same neighbors outside the module. A graph $G$ has \emph{modularwidth} at most $w$ if it satisfies at least one of the following conditions (i) $|V(G)| \leq w$, or (ii) there exists a partition of $V(G)$ into at most $w$ sets $V_1$, $V_2$, $\ldots$, $V_r$, such that $G[V_i]$ has modularwidth at most $w$ and $V_i$ is a module in $G$, for all $i \in [r]$. We will use $\textsf{mw}(G)$ to denote the minimum $w$ for which $G$ has modularwidth at most $w$. Note that there is a polynomial-time algorithm which, given a graph $G$ produces a non-trivial partition of $V(G)$ into at most $\textsf{mw}(G)$ modules~\cite{DBLP:conf/caap/CournierH94}, we call such a (non-trivial) partition a \emph{modular decomposition} of the graph. It is also not hard to see that deleting vertices cannot increase the modularwidth of a graph. This parameter has already been considered in the combinatorial reconfiguration framework; Belmonte etval.~\cite{DBLP:journals/algorithmica/BelmonteHLOO20} show that the \textsc{Independent Set Reconfiguration} problem is fixed-parameter tractable on graphs of bounded modularwidth under all three models (addition/removal, jumping, and sliding).

\subsubsection{Treedepth}

We start by showing that \textsc{Shortest Path Reconfiguration} is fixed-parameter tractable when parameterized by the treedepth of the input graph. Our proof is identical to the proof of Demaine et al.~\cite{DemaineEHJLUU19} for a slightly different path reconfiguration problem but we include it here for completeness. 

\begin{definition}[\cite{DemaineEHJLUU19}]
Given a graph $G$ and a vertex set $S$, an \emph{$S$-flap} is a subset $X \subseteq V(G)$ such that $X \cap S = \emptyset$ and there are no edges from $X$ to $V(G) \setminus S \cup X$. We say that two S-flaps $X$ and $Y$ are \emph{equivalent} when the induced subgraphs $G[S \cup X]$ and $G[S \cup Y]$ are isomorphic, by an isomorphism that reduces to the identity mapping on $S$.
\end{definition}

An essential observation in proving the next lemma is that a path of length at most $k$ has at most $k + 1$ vertices, and any two vertices in distinct flaps must be separated by at least one vertex of $S$. This implies that for any graph $G$ and any vertex set $S$, a path of length $k$ can include vertices from at most
$\lceil (k - 1)/2 \rceil$ $S$-flaps of $G$.

\begin{lemma}[\cite{DemaineEHJLUU19}]\label{lem:td-reduce}
Assume we are given an instance of \textsc{Shortest Path Reconfiguration} for paths of length $k$ in a graph $G$, and that $G$ contains a subset $S$ that is disjoint from the source and target shortest paths and has more than
$\lceil (k + 1)/2 \rceil$ pairwise equivalent $S$-flaps all disjoint from the source and target. Then, we can construct an equivalent and smaller instance by removing all but $\lceil (k + 1)/2 \rceil$ of these equivalent $S$-flaps.
\end{lemma}

\begin{theorem}\label{thm:fpt-td}
\textsc{Shortest Path Reconfiguration} (in either the token jumping or the token sliding model) is fixed-parameter tractable when parameterized by the treedepth of the input graph, i.e., $\textsf{td}(G)$. 
\end{theorem}

\begin{proof}
The proof consists of an algorithm that reduces an instance to an equivalent one (Using Lemma~\ref{lem:td-reduce}) such that the size of the reduced instance is a function only of the treedepth of the input graph. The problem can then be solved by
a brute-force search on the resulting smaller instance. We assume without loss of generality that we already
have a tree decomposition $T$ of depth $d$, as it is fixed-parameter tractable to find such a decomposition when one is not already given~\cite{DBLP:books/daglib/0030491}. Recall that, for graphs of treedepth $d$, the length $k$ of the paths being
reconfigured can be at most $2^{d+1} - 2$.
We apply Lemma~\ref{lem:td-reduce} in a sequence of stages so that, after stage $i$, for all vertices at height $i$ in $T$ the number of children is bounded by a function of $d$.  

As a base case, for stage $0$, all vertices at height $0$ in $T$ have $0$ children, since they
are the leaves of $T$. Therefore, suppose by induction on $i$ that all vertices at height less than $i$ in $T$ have a bounded number of children.
For a given vertex $v$ at height $i$, let $S_v$ be the set of ancestors of $v$ in $T$ (including $v$). Then, for each
child $w$ of $v$ in $T$, let $X_w$ be the set of descendants of $w$ (including $w$ itself). Then $X_w$ is an $S_v$-flap, because $S_v$ includes all of its ancestors in $T$ and it can have no edges to vertices that are not ancestors in $T$. If we label
each vertex in $T$ by the set of heights of its adjacent ancestors, then the isomorphism type of $G[S_v \cup X_w]$ is determined by these labels, so two children $u$ and $w$ of $T$ have equivalent $S_v$-flaps whenever they correspond to isomorphic labeled subtrees of $W$. Trees of bounded size with a bounded number of label values can have
a bounded number of isomorphism types, so there are a bounded number of equivalence classes of $S_v$-flaps
among the sets $W_x$. Within each equivalence class, we apply Lemma~\ref{lem:td-reduce} to reduce the number of flaps within that equivalence class to a number bounded by the treedepth. After doing so, the vertices of $T$ at height $i$ have a bounded number of children, completing the induction proof. 
\end{proof}
\subsubsection{Cluster deletion number}

We now move on to the parameterization by cluster deletion number. First, we state the following proposition which is immediate in a graph of bounded cluster deletion number. 

\begin{proposition}\label{prop:fpt-cd}
Let $G$ be a graph where $\textsf{cd}(G) \leq r$. Let $C$ denote the set whose deletion leaves a disjoint union of complete graphs. Then, any shortest path in $G$ contains at most $3r$ vertices in total and at most two vertices from each complete graph in $G - C$. 
\end{proposition}

\begin{proof}
Recall that $|C| \leq \textsf{cd}(G) \leq r$. Any shortest path in $G$ can contain at most two vertices from each complete subgraph of $G - C$. Moreover, if a shortest path contains all vertices of $C$ then each vertex in $C$ can contribute to at most two vertices from some complete graph, as needed.     
\end{proof}

\begin{theorem}\label{thm:fpt-cd}
\textsc{Shortest Path Reconfiguration} (in either the token jumping or the token sliding model) is fixed-parameter tractable when parameterized by the cluster deletion number of the input graph, i.e., $\textsf{cd}(G)$. 
\end{theorem}

\begin{proof}
By Proposition~\ref{prop:fpt-cd}, we can assume that the length of the shortest paths is at most $3 \times \textsf{cd}(G)$. 

Let $C \subseteq V(G)$ denote the set of size $\textsf{cd}(G)$ whose deletion leaves a disjoint union of complete graphs. We denote those cliques by $K_1$, $K_2$, $\ldots$, $K_q$. We now classify each vertex in $V(G) \setminus C$ based on its neighborhood in $C$. That is, we say $u$ and $v$ belong to the same class and have the same type whenever $N(u) \cap C = N(v) \cap C$. Next, with the exception of the vertices of $P$ and $Q$, whenever a clique contains more than one vertex of the same type we can delete all except one vertex. Hence, each clique eventually contains at most $2^{\textsf{cd}(G)}$ vertices.

It remains to bound the number of cliques in $G - C$. To do so, we define the type of a clique as the set of types of its vertices. Excluding the cliques that intersect with $P$ and $Q$, whenever we have more than $3 \times \textsf{cd}(G)$ cliques of the same type we can just retain $3 \times \textsf{cd}(G)$ cliques of that type. 
\end{proof}
\subsubsection{Modular width}
We now show that the \textsc{Shortest Path Reconfiguration} problem is fixed-parameter tractable when parameterized by modularwidth. Before doing so, we prove the following (straightforward but needed) proposition~\cite{DBLP:journals/corr/abs-0912-1457}. 

\begin{proposition}\label{prop:path-length-mw}
If a graph $G$ has modularwidth at most $w$, i.e., $\textsf{mw}(G) \leq w$, then any shortest path in $G$ has length at most $w$. Moreover, if $|V(G)| > w$ and $V_1$, $V_2$, $\ldots$, $V_r$, $r \leq w$, is a modular decomposition of $G$ then any shortest path of length at least three is either fully contained in some $G[V_i]$ or consists of at most one vertex from each $V_i$, for $i \in [r]$.
\end{proposition}

\begin{proof}
Let $P$ be a shortest path in $G$ between vertices $s,t \in V(G)$. Recall that a module is a set of vertices that have the same neighbors outside the module and $\textsf{mw}(G) \leq w$ if either $|V(G)| \leq w$ or there exists a partition of $V(G)$ into at most $w$ sets $V_1$, $V_2$, $\ldots$, $V_r$, such that $\textsf{mw}(G[V_i]) \leq w$ and $V_i$ is a module in $G$, for all $i \in [r]$. 

Note that if $|V(G)| \leq w$ then we are done. Hence, we assume that there exists a partition of $V(G)$ into at most $w$ (non-empty) module $V_1$, $V_2$, $\ldots$, $V_r$ such that for each module we can recursively apply the definition of modularwidth. 

Assume that $|V(P) \cap V_i| \geq 2$, for some $i \in [r]$, and $|V(P) \cap V_j| \geq 1$, for some $j \in [r]$ such that $i \neq j$. Then, we claim that the path $P$ cannot be a shortest path (unless it has length two); this indeed follows from the definition of modules since all vertices in $V_i$ have the same neighbors outside of $V_i$. Hence, whether the shortest path enters and leaves $V_i$ twice or more or visits two or more consecutive vertices in $V_i$ (before leaving) then we get a contradiction to the fact that $P$ is shortest. This implies that a shortest path in $G$ is either fully included in some $V_i$ or uses at most one vertex from each $V_i$. In the latter case, we get a path of length at most $w$ and we are done. In the former case, since $\textsf{mw}(G[V_i]) \leq w$, we can apply the same reasoning recursively. 
\end{proof}

\begin{theorem}\label{thm:fpt-mw}
\textsc{Shortest Path Reconfiguration} (in either the token jumping or the token sliding model) is fixed-parameter tractable when parameterized by the modularwidth of the input graph, i.e., $\textsf{mw}(G)$. 
\end{theorem}

\begin{proof}
Let $(G,s,t,P,Q)$ be an instance of \textsc{Shortest Path Reconfiguration}, where $\textsf{mw}(G) = w$ and $k \leq w$ denotes the length of $P$ and $Q$. 

If $|V(G)| \leq w$ or $k \leq 2$ then we can solve the problem by a brute-force search on the instance, i.e., we can construct a graph containing one vertex for each shortest path between $s$ and $t$ and then connect two vertices wherever the reconfiguration rule applies.

If $V(G) > w$, we first run the polynomial-time algorithm which, given a graph $G$, produces a non-trivial partition of $V(G)$ into at most $r \leq w$ sets $V_1$, $V_2$, $\ldots$, $V_r$, such that $\textsf{mw}(G[V_i]) \leq w$ and $V_i$ is a module in $G$, for all $i \in [r]$. We know, from Proposition~\ref{prop:path-length-mw}, that any shortest path in $G$ is either fully contained in some $G[V_i]$ or consists of exactly one vertex from each $V_i$, for $i \in [r]$. 

If both $s$ and $t$ belong to the same $V_i$, $i \in [r]$, then either the distance between them is at most two or every shortest path between them is fully contained in $G[V_i]$. To see why, assume that there exists a shortest path from $s$ to $t$ that is not fully contained in $G[V_i]$ and let $v$ be the first vertex along this path that is not in $V_i$. By the definition of modules, both $s$ and $t$ must be adjacent to $v$ which implies that the distance between them is at most two. The case of distance at most two is handled by the brute-force search. So we can assume that the distance between $s$ and $t$ is at least three and then we solve the problem recursively~in~$G[V_i]$. 

Finally, assume that $s$ belongs to $V_i$ and $t$ belongs to $V_j$, where $i \neq j$. We show that for any two shortest paths $W$ and $W'$ from $s$ to $t$ the following condition must hold. All the modules that $W$ and $W'$ have in common must appear at the same positions in $W$ and $W'$, respectively. 
Let $V^W_0 = V_i$, $V^W_1$, $\ldots$, $V^W_k = V_j$ denote the modules containing the vertices of $W$ (recall that we have one vertex per module). Similarly, let $V^{W'}_0 = V_i$, $V^{W'}_1$, $\ldots$, $V^{W'}_k = V_j$ denote the modules containing the vertices of $W'$. We now show that whenever $V^W_x = V^{W'}_y$ then $x = y$. In other words, whenever the same module intersects with both $W$ and $W'$ (in exactly one vertex which could be different for $W$ and $W'$) then it must be at the same position along the path. Assume, towards a contradiction, that $V^W_x = V^{W'}_y$ and $x < y$ (the proof of the case $x > y$ is analogous). Since $x < y$ it must be the case that the distance from $V_x$ to $t$ in $W'$ is less than the distance from $V_x$ to $t$ in $W$. However, this allows us to create a shortest path $W''$ which visits the modules $V^{W}_0 = V_i$, $\ldots$, $V^{W}_x = V^{W'}_y$, $V^{W'}_{y + 1}$, $V^{W'}_k = V_j$. This contradicts our assumption that both $W$ and $W'$ are shortest paths between $s$ and $t$. Note that by the definition of modules, whenever a module intersects with any shortest path between $s$ and $t$ then all the vertices of the module must be at the same distance from $s$ (and $t$). Hence, we can label each module using its distance from $s$ and safely delete all the modules which do not contain a vertex belonging to some shortest path between $s$ and $t$, Moreover, if one vertex of a module belongs to some shortest path between $s$ and $t$ then the same is true for all vertices of the module, i.e., we can replace one vertex of the module by any other vertex of the module and still maintain a shortest path (since all vertices of a module have the same neighbors outside the module). 

It remains to show how we can solve the problem given all of the above, For the token jumping model we can simply create a new auxiliary graph $H$ as follows. We add $P$ and $Q$ to $H$ as well as any vertex belonging to some shortest path from $s$ to $t$. We omit/delete edges between vertices of $H$ that are at the same distance from $s$ and $t$ and add/create edges between two vertices of $H$ whenever they are adjacent in $G$ and are not at the same distance from $s$ and $t$ in $G$. Finally, as long as we can find two vertices $u \in V(H) \setminus \{V(P) \cup V(Q)\}$ and $v$ in $H$ that are (false) twins, we delete $u$ and retain $v$; recall that vertices are said to be (false) twins whenever they have the same neighborhood and are non-adjacent. 
The correctness of this step follows from the fact that when we delete edges between vertices at the same distance from $s$ and $t$ we are left with modules that induce independent sets, which means that all vertices in a module are (false) twins and are equivalent. 
In the new instance $(H,s,t,P,Q)$ of \textsc{Shortest Path Reconfiguration} we can again solve the problem by a brute-force search since $|V(H)| \leq 4\textsf{mw}(G)$. 

The situation is slightly more complicated in the token sliding model since we need to maintain more information about connectivity within each module $V_x$, and in particular between the modules that are at the same distance from $s$ (and $t$). Recall that there are only two possible ways a pair of modules can share edges; either all edges are present or none. Therefore, we now create the graph $H$ as follows. We first add $P$ and $Q$ to $H$ as well as any vertex  belonging to some shortest path from $s$ to $t$. We add/create edges between two vertices of $H$ whenever they are adjacent in $G$ (regardless of distance to $s$). Consider a shortest sequence of slides from $P$ to $Q$. It is not hard to see that we can assume that no token slides within a module and then out of the module. Similarly, no token slides from one module to another and then slides inside the latter module; this follows from the definition of modules and the fact that we assume a shortest sequence of slides. Both types of slides can be skipped while preserving a valid (shorter) reconfiguration sequence. Given a module $V_x$, we contract each connected component in $H[V_x \setminus \{V(P) \cup V(Q)\}]$ to a single vertex. 

Since the number of connected components in each module can be unbounded we still cannot apply a brute-force search algorithm. To remedy the situation we proceed as follows. We guess which tokens will not leave their respective modules. For each such token on vertex $v$ in module $V_x$, we compute a shortest path from $v = V(P) \cap V_x$ to $w = V(Q) \cap V_x$, save it,  and then delete all vertices that are at the same distance as $v$ from $s$ and $t$. We also adjust $Q$ so that $w$ is replaced by $v$. The moves along the  $vw$-shortest path will be appended to the reconfiguration sequence computed in $H$. When $v$ and $w$ belong to different components of $V_x$ we simply ignore this guess and continue to the next. Now, assuming a correct guess, we know that in the instance we have, every slide of a token implies a change of module. Consider a module $V_y$ such that $V_y \cap \{V(P) \cup V(Q)\} = \emptyset$. We claim that retaining only one vertex of $V_y$ is enough. To see why, note that every move of a token into $V_y$ or out of $V_y$ can be done using any vertex of $V_y$ (all the vertices of a module have the same neighbors outside the module). Since $V_y$ is not the initial or final destination of our token, and no slides happen inside $V_y$ the claim follows. 
Using similar arguments, it follows that for a module $V_y$ such that $V_y \cap \{V(P) \cup V(Q)\} \neq \emptyset$ retaining only $|V_y \cap \{V(P) \cup V(Q)\}|$ vertices is enough, i.e., we always retain the vertex of $P$ and the vertex of $Q$ when reducing the number of independent vertices in modules. 

Having bounded the number of vertices inside each module to two, we can now solve the problem via brute force as in the token jumping case (with the extra edges between vertices at the same distance). In other words, we have $|V(H)| \leq 4\textsf{mw}(G)$, as needed to complete the proof.
\end{proof}

\subsubsection{Feedback vertex set number}
We conclude this section by proving that on graphs of bounded feedback vertex set number, we can always assume that the paths in an instance of \textsc{Shortest Path Reconfiguration} are of length at most $4 \times \textsf{fvs}(G)$. 

\begin{lemma}\label{lem:fvs-path-length}
Let $(G,s,t,P,Q)$ be an instance of \textsc{Shortest Path Reconfiguration} (in either the token jumping or the token sliding model), where $\textsf{fvs}(G) \leq r$. Then, we can always assume that either $k \leq 4 \times \textsf{fvs}(G)$ or $(G,s,t,P,Q)$ is a no-instance. 
\end{lemma}

\begin{proof}
We first run a breadth-first search traversal of the input graph starting from $s$ and we label each vertex that we explore using its distance from $s$; $s$ receives label $0$ and $t$ receives label $k$. We can safely delete all vertices that are unlabeled or labelled by $k$ or more (except for $t$). Now, starting from $t$, we traverse the labelled vertices in reverse (by running another breadth-first search) and maintain only those vertices that belong to some shortest path between $s$ and $t$. 
We now create level sets starting with $L_0 = \{s\}$, $L_1$ containing all vertices at distance one from $s$, $\ldots$, $L_i$ containing all vertices at distance $i$ from $s$, and finally $L_k = \{t\}$. 

Assume that $k \geq 4 \times \textsf{fvs}(G) + 1$ and let $F$ be a feedback vertex set of $G$ of size $\textsf{fvs}(G)$. Hence, there must exist at least four consecutive level sets $L_i$, $L_{i+1}$, $L_{i+2}$, and $L_{i+3}$ such that $F \cap (L_i \cup L_{i+1} \cup L_{i+2} \cup L_{i+3}) = \emptyset$. This implies that $G[L_i \cup L_{i+1} \cup L_{i+2} \cup L_{i+3}]$ must be a forest and therefore the tokens in $L_{i+1}$ and $L_{i+2}$ cannot move. If $P$ and $Q$ do not agree on the vertices in those two level sets then we have a no-instance. Otherwise, let $u$ and $v$ denote the vertices in $L_{i+1}$ and $L_{i+2}$, respectively. We delete $L_{i+1}$ and $L_{i+2}$ from the graph and replace them by a single vertex $w$ such that $w$ is connected to all neighbors of $u$ in $L_i$ and all neighbors of $v$ in $L_{i+3}$. It is not hard to see that we obtain an equivalent instance where the length of the shortest path is reduced by one, as needed to complete the proof. 
\end{proof}

\bibliographystyle{plain}
\bibliography{refs}
\end{document}